\pdfoutput=1
\documentclass[]{article}
\usepackage{times}
\usepackage{fullpage}
\usepackage{url}
\usepackage{verbatim} 
\usepackage{graphicx}
\usepackage{caption}
\usepackage{subcaption}
\usepackage{multirow}
\usepackage[vlined,ruled]{algorithm2e}
\usepackage{xspace}
\usepackage{amsmath}
\usepackage{amssymb}
\usepackage{Definitions}
\usepackage{empheq}
\usepackage{lipsum}
\usepackage{color}
\usepackage{xspace}
\usepackage{enumitem}
\usepackage{authblk}

\begin{document}

\title{Shaping Social Activity by Incentivizing Users}

\author[1]{Mehrdad Farajtabar}
\author[1]{Nan Du}
\author[2]{Manuel Gomez Rodriguez}
\author[3]{Isabel Valera}
\author[1]{Hongyuan Zha}
\author[1]{Le Song}

\affil[1]{Georgia Institute of Technology, \{mehrdad,dunan\}@gatech.edu, \{zha,lsong\}@cc.gatech.edu}
\affil[2]{Max Plank Institute for Intelligent Systems, manuelgr@tuebingen.mpg.de}
\affil[3]{University Carlos III in Madrid, ivalera@tsc.uc3m.es}

\date{}

\begin{small}
\maketitle
\end{small}


\begin{abstract}
Events in an online social network can be categorized roughly into \emph{endogenous} events,
where users just respond to the actions of their neighbors within the network, or \emph{exogenous} events, where users take actions due to drives external to the network.
How much external drive should be provided to each user, such that the network activity can be steered towards a target state?
In this paper, we model social events using multivariate Hawkes processes, which can capture both endogenous and exogenous event intensities, and derive a time dependent linear relation between the intensity of
exogenous events and the overall network activity.
Exploiting this connection, we develop a convex optimization framework for determining the required level of external drive in order for the network to reach a
desired activity level.
We experimented with event data gathered from Twitter, and show that our method can steer the activity of the network more accurately than alternatives.
\end{abstract}

\section{Introduction}
\label{sec:intro}

\setlength{\abovedisplayskip}{3pt}
\setlength{\abovedisplayshortskip}{1pt}
\setlength{\belowdisplayskip}{3pt}
\setlength{\belowdisplayshortskip}{1pt}
\setlength{\jot}{2pt}

\setlength{\floatsep}{2ex}
\setlength{\textfloatsep}{2ex}

Online social platforms routinely track and record a large volume of event data, which may correspond to the usage of a service (\eg, url shortening service, bit.ly).
These events can be categorized roughly into \emph{endogenous} events,
where users just respond to the actions of their neighbors within the network, or \emph{exogenous} events, where users take actions due to drives external to the network.
For instance, a user's tweets may contain links provided by bit.ly, either due to his forwarding of a link from his friends, or due to his own initiative to use the service to create a new link.

Can we model and exploit these data to steer the online community to a desired activity level?
Specifically, can we drive the overall usage of a service to a certain level (\eg, at least twice per day per user) by incentivizing a small number of users to take more initiatives?
What if the goal is to make the usage level of a service more homogeneous across users?
What about maximizing the overall service usage for a target group of users?
Furthermore, these \emph{activity shaping} problems need to be addressed by taking into account budget constraints, since incentives are usually provided in the form of monetary or credit rewards.

Activity shaping problems are significantly more challenging than traditional influence maximization problems, which aim to identify a set of
users, who, when convinced to adopt a product, shall influence others in the network and trigger a large cascade of adoptions~\cite{KemKleTar03, RicDom02}.
First, in influence maximization, the state of each user is often assumed to be binary, either adopting a product or not~\cite{KemKleTar03, CheWanYan09, RodSch12, DuSonRodZha13}.
However, such assumption does not capture the recurrent nature of product usage, where the frequency of the usage matters.
Second, while influence maximization methods identify a set of users to provide incentives, they do not typically provide a quantitative prescription on how
much incentive should be provided to each user.
Third, activity shaping concerns about a larger variety of target states, such as minimum activity requirement and homogeneity of activity, not just activity maximization.

In this paper, we will address the activity shaping problems using multivariate Hawkes processes~\cite{Liniger2009}, which can model both endogenous and exogenous recurrent social events, and were shown to be a good fit for such data in a number of recent works (\eg,~\cite{BluBecHelKat12, ZhoZhaSon13, ZhoZhaSon13b, IwaShaGha13, LinAdaRya14, ValGomGum14}).
More importantly, we will go beyond model fitting, and derive a novel predictive formula for the overall network activity given the intensity of exogenous events in individual users, using a connection between the processes and branching processes~\cite{DobCarBenNew04, Rasmussen13, VeeSch08, ZuaOgaVer02}.
%
%
Based on this relation, we propose a convex optimization framework to address a diverse range of activity shaping problems given budget constraints. Compared to previous methods for influence maximization, our framework can provide more fine-grained control of network activity, not only steering the network to a desired steady-state activity level but also do so in a time-sensitive fashion. For example, our framework allows us to answer complex time-sensitive queries, such as, which users should be incentivized, and by
how much, to steer a set of users to use a product twice per week after one month?

In addition to the novel framework, we also develop an efficient gradient based optimization algorithm, where the matrix exponential needed for gradient computation is approximated using the truncated Taylor series expansion~\cite{AlmHig11}. This algorithm allows us to validate our framework in a variety of activity shaping tasks and scale up to networks with tens of thousands of nodes. We also conducted experiments on a network of 60,000 Twitter users and more than 7,500,000 uses of a popular url shortening service. Using held-out data, we show that our algorithm can shape the network behavior much more accurately.

\section{Modeling Endogenous-Exogenous Recurrent Social Events}
\label{sec:formulation}
%
We model the events generated by $m$ users in a social network as a $m$-dimensional counting process $\Nb(t) = (N_1(t), N_2(t), \ldots, N_m(t))^\top$, where $N_i(t)$ records the total
number of events generated by user $i$ up to time $t$.
Furthermore, we represent each event as a tuple $(u_i, t_i)$, where $u_i$ is the user identity and $t_i$ is the event timing.
Let the history of the process up to time $t$ be $\Hcal_t:=\cbr{(u_i,t_i)\, |\, t_i \leqslant t}$, and $\Hcal_{t-}$ be the history until just before time $t$.
Then the increment of the process, $d\Nb(t)$, in an infinitesimal window $[t,t+dt]$ is parametrized by the intensity $\lambdab(t)=(\lambda_1(t), \ldots, \lambda_m(t))^\top\geqslant 0$,~\ie,
\begin{align}
  \EE[d\Nb(t) | \Hcal_{t-}] = \lambdab(t) \, dt.
\end{align}
Intuitively, the larger the intensity $\lambdab(t)$, the greater the likelihood of observing an event in the time window $[t, t+d t]$.
For instance, a Poisson process in $[0,\infty)$ can be viewed as a special counting process with a constant intensity function $\lambdab$, independent of time and history. To model the presence of both endogenous and exogenous events, we will decompose the intensity into two terms
\begin{align}
  \underbrace{\lambdab(t)}_{\text{overall event intensity}} \quad =\quad \underbrace{\lambdab^{(0)}(t)}_{\text{exogenous event intensity}} \quad +\quad  \underbrace{\lambdab^{\ast}(t)}_{\text{endogenous event intensity}},
\end{align}
where the exogenous event intensity models drive outside the network, and the endogenous event intensity models interactions within the network. We assume that hosts of social platforms can potentially drive up or down the exogenous events intensity by providing incentives to users; while endogenous events are generated due to users' own interests or under the influence of network peers, and the hosts do not interfere with them directly. The key questions in the activity shaping context are how to model the endogenous event intensity which are realistic to recurrent social interactions, and how to link the exogenous event intensity to the endogenous event intensity. We assume that the exogenous event intensity is independent of the history and time,~\ie,~$\lambdab^{(0)}(t) = \lambdab^{(0)}$.

\subsection{Multivariate Hawkes Process}
%
Recurrent endogenous events often exhibit the characteristics of self-excitation, where a user tends to repeat what he has been doing recently, and mutual-excitation, where a user simply follows what his neighbors are doing due to peer pressure. These social phenomena have been made analogy to the occurrence of earthquake~\cite{MarLen08} and the spread of epidemics~\cite{YanZha13}, and can be well-captured by multivariate Hawkes processes~\cite{Liniger2009} as shown in a number of recent works (\eg,~\cite{BluBecHelKat12, ZhoZhaSon13, ZhoZhaSon13b, IwaShaGha13, LinAdaRya14,ValGomGum14}).

More specifically, a multivariate Hawkes process is a counting process who has a particular form of intensity. More specifically, we assume that the strength of influence between users is parameterized by a sparse nonnegative \emph{influence matrix}
$\Ab = (a_{uu'})_{u,u' \in [m]}$, where $a_{uu'}> 0$ means user $u'$ directly excites user $u$.  We also allow $\Ab$ to have nonnegative diagonals to model self-excitation of a user.
Then, the intensity of the $u$-th
dimension is
\begin{align}
  \lambda_{u}^{*}(t)= \sum\nolimits_{i: t_i < t} a_{uu_i}\, g(t-t_i) = \sum\nolimits_{u' \in [m]} a_{uu'} \int_{0}^t g(t-s) \, dN_{u'}(s),
  \label{eq:hawkes-intensity}
\end{align}
where $g(s)$ is a nonnegative kernel function such that $g(s)=0$ for $s\leq 0$ and $\int_{0}^{\infty} g(s)\,ds < \infty$; the second equality is obtained by grouping events according to users and use the fact that $\int_{0}^t g(t-s) \, dN_{u'}(s)=\sum_{u_i=u',t_i<t} g(t-t_i)$.
Intuitively, $\lambda_{u}^{*}(t)$ models the propagation of peer influence over the network --- each event $(u_i,t_i)$ occurred in the neighbor of
a user will boost her intensity by a certain amount which itself decays over time. Thus, the more frequent the events occur in the user'{}s neighbor, the more
likely she will be persuaded to generate a new event.

For simplicity, we will focus on an exponential kernel, $g(t-t_i) =  \exp(-\omega (t-t_i))$ in the reminder of the paper. However, multivariate Hawkes processes and the branching processed explained in next section is independent of the kernel choice and can be extended to other kernels such as power-law, Rayleigh or any other long
tailed distribution over nonnegative real domain.
Furthermore, we can rewrite equation~\eq{eq:hawkes-intensity} in vectorial format
\begin{align}
  \lambdab^{*}(t) = \int_{0}^{t} \Gb(t-s)\, d \Nb(s),
  \label{eq:hawkes-intensity-conv}
\end{align}
by defining a $m\times m$ time-varying matrix $\Gb(t) = (a_{uu'} g(t))_{u, u'\in [m]}$.
Note that, for multivariate Hawkes processes, the intensity, $\lambdab(t)$, itself is a random quantity, which depends on the history $\Hcal_t$. We denote the expectation of the intensity with respect to history as
\begin{align}
  \label{eq:expected_intensity}
  \mub(t) := \EE_{\Hcal_{t-}}\sbr{\lambdab(t)}
\end{align}


\begin{figure*} [!t]
        \centering
        \begin{tabular}{cc}
          \includegraphics[width=0.16\textwidth]{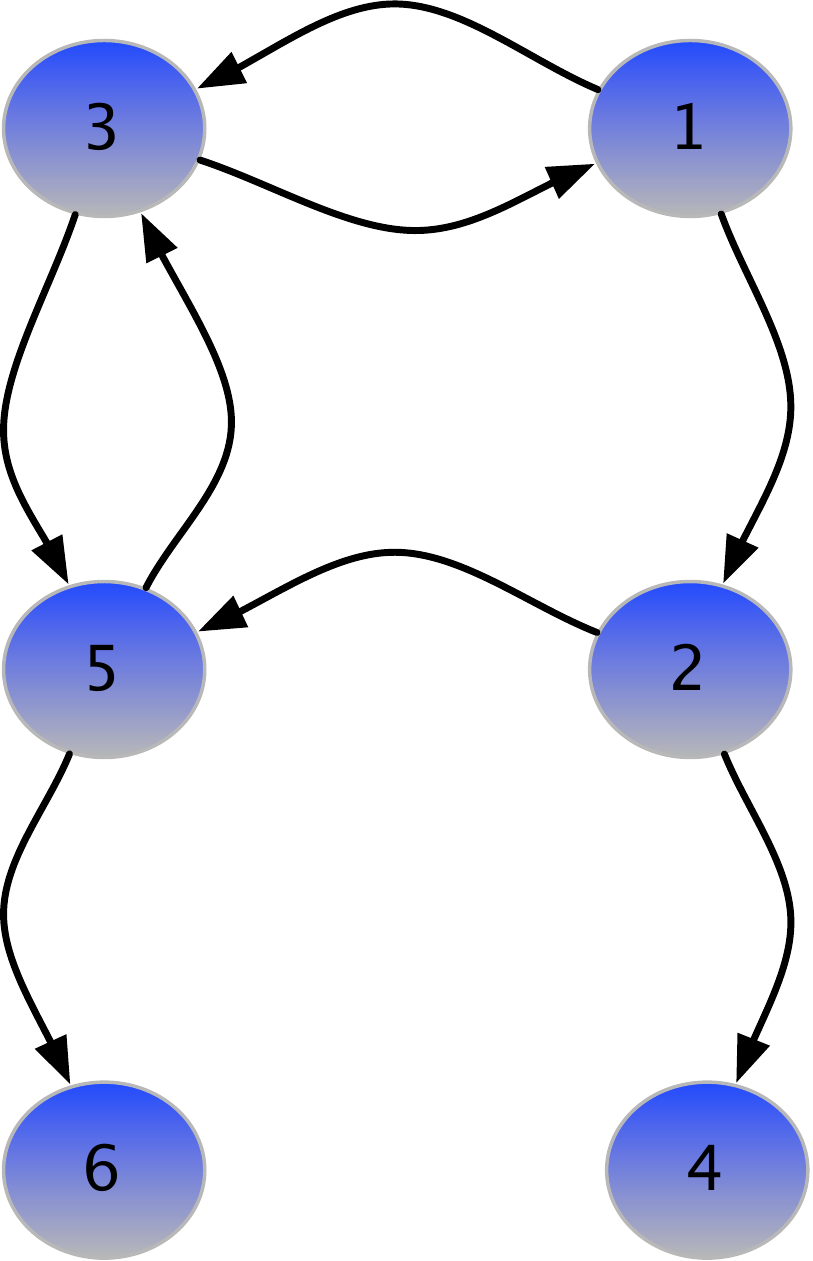} &
          \includegraphics[width=0.45\textwidth]{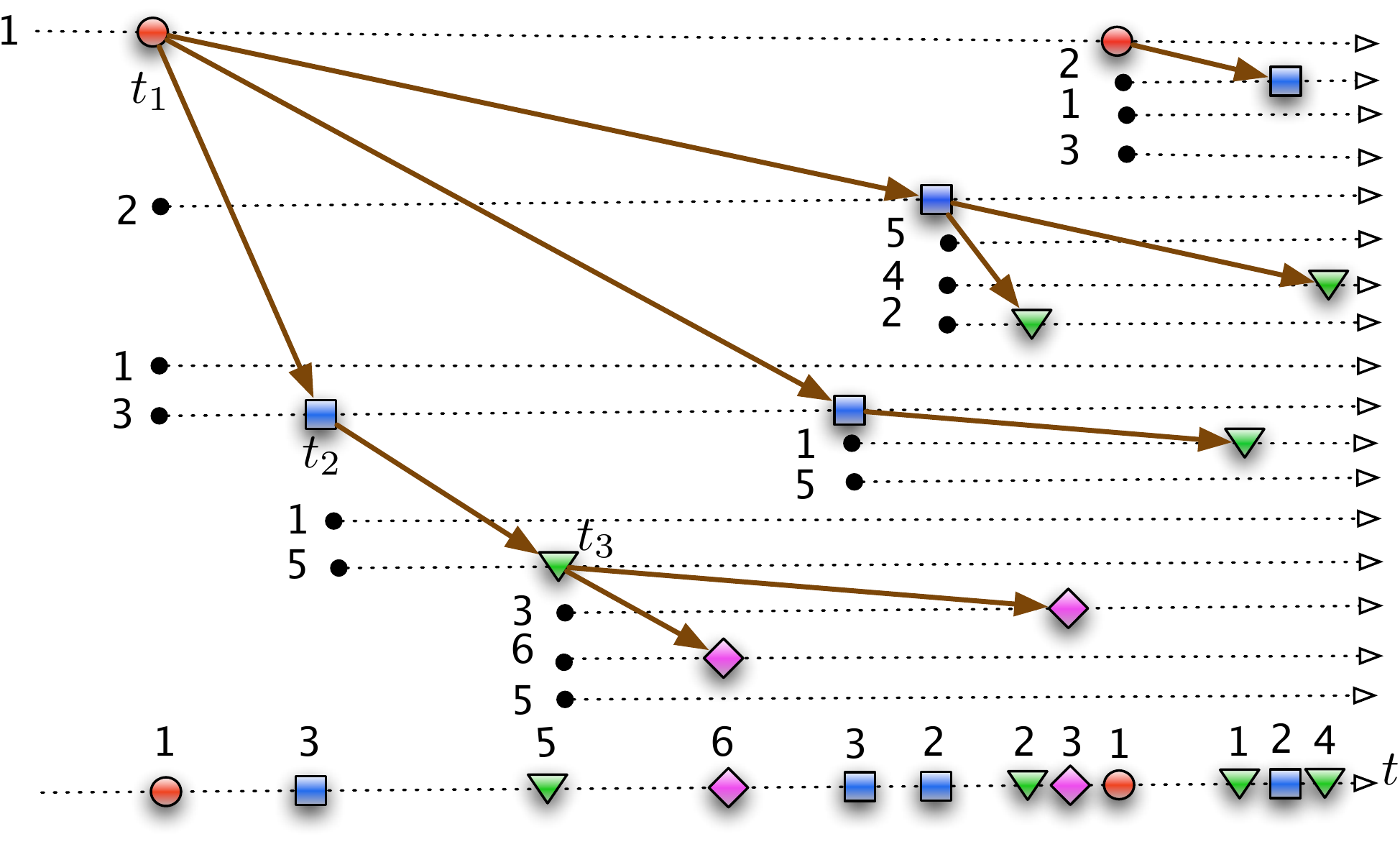} \\
          (a) An example social network & (b) Branching structure of events
        \end{tabular}
        \caption{(a) an example social network where each directed edge indicates that the target node \emph{follows}, and can be influenced by, the source node. The activity in this network is modeled using Hawkes processes, which result in branching structure of events in (b).
        Each exogenous event is the root node of a branch (\eg,~top left most red circle at $t_1$), and it occurs due to a user'{}s own initiative; and each event can trigger
        one or more endogenous events (blue square at $t_2$).
        The new endogenous events can create the next generation of endogenous events (green triangles at $t_3$), and so forth.
        The social network in (a) will constrain the branching structure of events in (b), since an event produced by a user (\eg,
        user $1$) can only trigger endogenous events in the same user or one or more of her followers (\eg, user $2$ or user $3$).
        }
        \label{fig:branching}
\end{figure*}

\subsection{Connection to Branching Processes}
\label{sec:branching}

A branching process is a Markov process that models a population in which each individual in generation $k$ produces some random number of individuals in generation $k + 1$, according some distribution~\cite{Harris02}. In this section, we will conceptually assign both exogenous events and endogenous events in the multivariate Hawkes process to levels (or generations), and associate these events with a branching structure which records the information on which event triggers which other events (see Figure~\ref{fig:branching} for an example). Note that this genealogy of events should be interpreted in probabilistic terms and may not be observed in actual data. Such connection has been discussed in Hawkes' original paper on one dimensional Hawkes processes~\cite{Hawkes71}, and it has recently been revisited in the context of multivariate Hawkes processes by~\cite{LinAdaRya14}.
%
%
The branching structure will play
a crucial role in deriving a novel link between the intensity of the exogenous events and the overall network activity.

More specifically, we assign all exogenous events to the zero-th generation, and record the number of such events as $\Nb^{(0)}(t)$. These exogenous events will trigger the first generation of endogenous events whose number will be recorded as $\Nb^{(1)}(t)$. Next these first generation of endogenous events will further trigger a second generation of endogenous events $\Nb^{(2)}(t)$, and so on and so forth. Then the total number of events in the network is the sum of the numbers of events from all generations
\begin{equation}
\label{eq:sum_point_process}
\Nb(t) = \Nb^{(0)}(t) +  \Nb^{(1)}(t) + \Nb^{(2)}(t) + \ldots
\end{equation}
Furthermore, denote all events in generation $k-1$ as $\Hcal_t^{(k-1)}$. Then, independently for each event $(u_i,t_i) \in \Hcal_t^{(k-1)}$ in generation $k-1$, it triggers a Poisson process in its neighbor $u$ independently with intensity $a_{uu_i} g(t-t_i)$. Due to the additivity of independent Poisson processes~\cite{Kingman92}, the intensity, $\lambda_u^{(k)}(t)$, of events at node $u$ and generation $k$ is simply the sum of conditional intensities of the Poisson processes triggered by all its neighbors,~\ie,
$
  \lambda_u^{(k)}(t) = \sum_{(u_i,t_i) \in \Hcal_t^{(k-1)}} a_{uu_i} g(t-t_i) = \sum_{u' \in [m]} \int_0^t g(t-s)\, d \Nb_{u'}^{(k-1)}(s).
$
Concatenate the intensity for all $u\in [m]$, and use the time-varying matrix $\Gb(t)$~\eq{eq:hawkes-intensity-conv}, we have
\begin{equation}
  \label{eq:sum_instant_intensity}
  \lambdab^{(k)}(t) = \int_{0}^t \Gb(t-s)\, d\Nb^{(k-1)}(s),
\end{equation}
where $\lambdab^{(k)}(t) = (\lambda_1^{(k)}(t),\ldots,\lambda_m^{(k)}(t))^\top$ is the intensity for counting process $\Nb^{(k)}(t)$ at $k$-th generation.
Again, due to the additivity of independent Poisson processes, we can decompose the intensity of $\Nb(t)$ into a sum of conditional intensities from different generation
\begin{equation}
  \label{eq:lambda_dec}
  \lambdab(t) = \lambdab^{(0)}(t) + \lambdab^{(1)}(t) + \lambdab^{(2)}(t)+\ldots
\end{equation}
Next, based on the above decomposition, we will develop a closed form relation between the expected intensity $\mub(t) = \EE_{\Hcal_{t-}}\sbr{\lambdab(t)}$ and the intensity, $\lambdab^{(0)}(t)$, of the exogenous events. This relation will form the basis of our activity shaping framework.

\section{Linking Exogenous Event Intensity to Overall Network Activity}

Our strategy is to first link the expected intensity $\mub^{(k)}(t):=\EE_{\Hcal_{t-}}[\lambdab^{(k)}(t)]$ of events at the $k$-th generation with $\lambdab^{(0)}(t)$, and then derive a close form for the infinite series sum
\begin{equation}
  \label{eq:mu_dec}
  \mub(t) = \mub^{(0)}(t) + \mub^{(1)}(t) + \mub^{(2)}(t)+\ldots
\end{equation}
Define a series of auto-convolution matrices, one for each generation, with $\Gb^{(\star 0)}(t) = \Ib$ and
\begin{equation}
\Gb^{(\star k)}(t) = \int_{0}^t \Gb(t-s)\, \Gb^{(\star k-1)}(s)\, ds = \Gb(t) \star \Gb^{(\star k-1)}(t)
\end{equation}
Then the expected intensity of events at the $k$-th generation is related to exogenous intensity $\lambdab^{(0)}$ by
\begin{lemma}
  \label{lem:br_intensity}
  $\mub^{(k)}(t) = \Gb^{(\star k)}(t)\,  \lambdab^{(0)}$.
\end{lemma}
Next, by summing together all auto-convolution matrices,
$$
  \Psib(t) := \boldsymbol{I}+  \Gb^{(\star 1)}(t) + \Gb^{(\star 2)}(t) + \ldots
$$
we obtain a linear relation between the expected intensity of the network and the intensity of the exogenous events,~\ie, $\mub(t)  = \Psib(t)  \lambdab^{(0)}$. The entries in the matrix $\Psib(t)$ roughly encode the ``influence'' between pairs of users.
More precisely, the entry $\Psib _{uv}(t)$ is the expected intensity of events at node $u$ due to a unit level of exogenous intensity at node $v$.
We can also derive several other useful quantities from $\Psib(t)$.
For example,  $\Psib_{\bullet v}(t) := \sum_u \Psib_{uv}(t)$ can be thought of as the overall influence user $v$ on has on all users. Surprisingly, for exponential kernel, the infinite sum of matrices results in a closed form using matrix exponentials. First, let~$\widehat{\cdot}$~denote the Laplace transform of a function, and we have the following intermediate results on the Laplace transform of $\Gb^{(\star k)}(t)$.
\begin{lemma}
  \label{lem:laplace}
  $\widehat{\Gb}^{(\star k)}(z) = \int_0^\infty \Gb^{(\star k)}(t) \, dt = \frac{1}{z}
  \cdot \frac{\Ab^k}{(z+\omega)^k}$
\end{lemma}
With Lemma~\ref{lem:laplace}, we are in a position to prove our main theorem below:
\begin{theorem}
  \label{theo:lin_rel}
  $
    \mub(t) = \Psib(t) \lambdab^{(0)} = \rbr{ e^{(\boldsymbol{A}-\omega \boldsymbol{I})t} +
    \omega (\boldsymbol{A}-\omega \boldsymbol{I})^{-1} ( e^{(\boldsymbol{A}-\omega \boldsymbol{I})t} - \boldsymbol{I} ) } \lambdab^{(0)}.
  $
\end{theorem}
%
Theorem~\ref{theo:lin_rel} provides us a linear relation between exogenous event intensity and the expected overall intensity at any point in time but not just stationary intensity. The significance of this result is that it allows us later to design a diverse range of convex programs to determine the intensity of the exogenous event in order to achieve a target intensity.

In fact, we can recover the previous results in the stationary case as a special case of our general result. More specifically, a multivariate Hawkes process is stationary if the spectral radius
\begin{equation}
  \boldsymbol{\Gamma} := \int_{0}^{\infty} \Gb(t)\,dt = \rbr{\int_{0}^{\infty}g(t)\,dt} \Bigl( a_{uu'} \Bigr)_{u,u'\in [m]} = \frac{\Ab}{\omega}
\end{equation}
is strictly smaller than 1~\cite{Liniger2009}. In this case, the expected intensity is $\mub = (\Ib - \boldsymbol{\Gamma})^{-1} \lambdab^{(0)}$ independent of the time.
We can obtain this relation from theorem~\ref{theo:lin_rel} if we let $t \rightarrow \infty$.
\begin{corollary}
  $\mub =  \left(\mathbf{I} - \boldsymbol{\Gamma}\right)^{-1} \lambdab^{(0)} = \lim_{t \rightarrow \infty} \Psib(t) \, \lambdab^{(0)}$.
\end{corollary}
Refer to Appendix \ref{append:proofs} for all the proofs.

\section{Convex Activity Shaping Framework}
\label{sec:proposed}
Given the linear relation between exogenous event intensity and expected overall event intensity, we now propose a convex optimization framework for a variety of activity shaping tasks. In all tasks discussed below, we will optimize the exogenous event intensity $\lambdab^{(0)}$ such that the expected overall event intensity $\mub(t)$ is maximized with respect to some concave utility $U(\cdot)$ in $\mub(t)$,~\ie,
\begin{equation}
	\label{eq:generalized-activity-maximization}
	\begin{array}{ll}
		\mbox{maximize}_{\mub(t),\lambdab^{(0)}} & U( \mub(t)) \\
		\mbox{subject to} & \mub(t)=\Psib(t) \lambdab^{(0)},\quad \cbb^\top \lambdab^{(0)} \leqslant C,\quad \lambdab^{(0)} \geqslant 0
	\end{array}
\end{equation}
where $\cbb = (c_1,\ldots,c_m)^\top \geqslant 0$ is the cost per unit event for each user and $C$ is the total budget. Additional regularization can also be added to $\lambdab^{(0)}$ either to restrict the number of incentivized users (with $\ell_0$ norm $\|\lambdab^{(0)}\|_0$), or to promote a sparse solution (with $\ell_1$ norm $\|\lambdab^{(0)}\|_1$, or to obtain a smooth solution (with $\ell_2$ regularization $\|\lambdab^{(0)}\|_2$).
We next discuss several instances of the general framework which achieve different goals (their constraints remain the same and hence omitted).

{\bf Capped Activity Maximization.}
In real networks, there is an upper bound (or a cap) on the activity each user can generate due to limited attention of a user. For example,
a Twitter user typically posts a limited number of shortened urls or retweets a limited number of tweets~\cite{GomGumSch14}. Suppose we know the upper bound, $\alpha_u$, on a user'{}s activity,~\ie, how much activity each user is willing to generate.
Then we can perform the following \emph{capped activity maximization} task
\begin{equation}
	\label{eq:const-average-activity-maximization}
	\begin{array}{ll}
		\mbox{maximize}_{\mub(t),\lambdab^{(0)}} & \sum_{u\in [m]} \min\cbr{\mu_u(t), \alpha_u} \\
	\end{array}
\end{equation}

{\bf Minimax Activity Shaping.}
Suppose our goal is instead maintaining the activity of each user in the network above a certain minimum level, or, alternatively make the user with the minimum activity as active as possible.
Then, we can perform the following \emph{minimax activity shaping} task
\begin{equation}
	\label{eq:minimax-activity-shaping}
	\begin{array}{ll}
		\mbox{maximize}_{\mub(t),\lambdab^{(0)}} & \min_u ~\mu_u(t)\\
 	\end{array}
\end{equation}

{\bf Least-Squares Activity Shaping.}
Sometimes we want to achieve a pre-specified target activity levels, $\vb$, for users.
For example, we may like to divide users into groups and desire a different level of activity in each group. Inspired by these examples, we can perform the following \emph{least-squares activity shaping} task
\begin{equation}
	\label{eq:least-squares-activity-shaping}
	\begin{array}{ll}
		\mbox{maximize}_{\mub(t),\lambdab^{(0)}} & -\|\Bb \mub(t)-\vb\|_2^2 \\
	\end{array}
\end{equation}
where $\Bb$ encodes potentially additional constraints (\eg, group partitions). Besides Euclidean distance, the family of Bregman divergences can be used to measure the difference between $\Bb \mub(t)$ and $\vb$ here. That is, given a function $f(\cdot):\RR^m \mapsto \RR$ convex in its argument, we can use $D(\Bb \mub(t) \| \vb):=f(\Bb \mub(t)) - f(\vb) - \inner{\nabla f(\vb)}{\Bb \mub(t) - \vb}$ as our objective function.

{\bf Activity Homogenization.} Many other concave utility functions can be used. For example, we may want to steer users activities to a more homogeneous profile. If we measure homogeneity of activity with Shannon entropy, then we can perform the following {activity homogenization} task
\begin{equation}
  \label{eq:activity-shaping}
  \begin{array}{ll}
    \mbox{maximize}_{\mub(t),\lambdab^{(0)}} & -\sum_{u \in [m]} \mu_u(t) \ln \mu_u(t)  \\
  \end{array}
\end{equation}

\section{Scalable Algorithm}
%
All the activity shaping problems defined above require an efficient evaluation of the instantaneous average intensity $\mub(t)$ at time $t$, which entails computing matrix 
exponentials to obtain $\boldsymbol{\Psi}(t)$.
In small or medium networks, we can rely on well-known numerical methods to compute matrix exponentials~\cite{GolVan12}.
However, in large networks with sparse graph structure $\Ab$, the explicit computation of $\boldsymbol{\Psi}(t)$ quickly becomes intractable.

Fortunately, we can exploit the following key property of our convex activity shaping framework: the instantaneous average intensity only depends on $\boldsymbol{\Psi}(t)$ through 
matrix-vector product operations. In particular, we start by using Theorem \ref{theo:lin_rel} to rewrite the multiplication of $\boldsymbol{\Psi}(t)$ and a vector $\vb$ as
$
\boldsymbol{\Psi}(t) \vb = e^{(\boldsymbol{A}-\omega \boldsymbol{I})t} \vb+
  \omega (\boldsymbol{A}-\omega \boldsymbol{I})^{-1} \left(  e^{(\boldsymbol{A}-\omega \boldsymbol{I})t}\vb - \vb \right)
$.
We then get a tractable solution by first computing $e^{(\boldsymbol{A}-\omega \boldsymbol{I})t}\vb$ efficiently, subtracting $\vb$ from it, and solving a sparse linear system of 
equations, $ (\boldsymbol{A}-\omega \boldsymbol{I}) x =  \left(  e^{(\boldsymbol{A}-\omega \boldsymbol{I})t}\vb - \vb \right)$, efficiently. 
The steps are illustrated in Algorithm~\ref{exp-algorithm}.
Next, we elaborate on two very efficient algorithms for computing the product of matrix exponential with a vector and for solving a sparse linear system of equations.

\begin{algorithm}[t]
\caption{Average Instantaneous Intensity}
\label{exp-algorithm}
\SetKwInOut{Input}{input}
\SetKwInOut{Output}{output}
\SetKw{KwRet}{return}
\SetKwComment{Comment}{}{}
\Input{$\Ab$, $\omega$, $t$, $\vb$}
\Output{$\Psib(t)\vb$}
$\vb_1 = e^{(\boldsymbol{A}-\omega \boldsymbol{I})t} \vb$  \Comment*[l]{\hspace{1cm}//Matrix exponential times a vector}
$\vb_2 = \vb_2 - \vb$; \\
$\vb_3 = (\boldsymbol{A}-\omega \boldsymbol{I})^{-1} \vb_2$ \Comment*[l]{\hspace{1cm}//Sparse linear systems of equation}
\KwRet{$\vb_1+ \omega \vb_3$};
\end{algorithm}

For the computation of the product of matrix exponential with a vector, we rely on the iterative algorithm by Al-Mohy et al.~\cite{AlmHig11}, which combines a scaling and squaring method 
with a truncated Taylor series approximation to the matrix exponential.
%

For solving the sparse linear system of equation, we use the well-known GMRES method~\cite{SaaSch86}, which is an Arnoldi process for constructing an $l_2 $-orthogonal basis of Krylov subspaces. 
The method solves the linear system by iteratively minimizing the norm of the residual vector over a Krylov subspace. In detail, consider the $n^{th}$ Krylov subspace for the problem $\Cb \xb = \bb$ 
as
$
\Kb_n = \operatorname{span}  \{ \bb, \Cb \bb, \Ab^2\bb, \ldots, \Cb^{n-1}\bb \}.
$
GMRES approximates the exact solution of  $\Cb \xb = \bb$ by the vector $\xb_n \in \Kb_n$ that minimizes the Euclidean norm of the residual $\rb_n = \Cb \xb_n - \bb$. Because the span consists of 
orthogonal vectors, the Arnoldi iteration is used to find an alternative basis composing rows of $\Qb_n$. Hence, the vector $\xb_n \in \Kb_n$ can be written as
$\xb_n = \Qb_n \yb_n$ with $\yb_n \in \mathbb{R}^n$. 
Then, $\yb_n$ can be found by minimizing the Euclidean norm of the residual $\rb_n = \tilde{\Hb}_n \yb_n - \beta \eb_1$, where $\tilde{\Hb}_n$ is the Hessenberg matrix produced in the Arnoldi 
process, $\eb_1 = (1,0,0,\ldots,0)^T$ is the first vector in the standard basis of $\mathbb{R}^{n+1}$, and $\beta = \|\bb-\Cb \xb_0\| $. 
Finally, $\xb_n$ is computed as $\xb_n = \Qb_n \yb_n$. 
The whole procedure is repeated until reaching a small enough residual.

Perhaps surprisingly, we will now show that it is possible to compute the gradient of the objective functions of all our activity shaping problems using the algorithm developed above for 
computing the average instantaneous intensity. We only need to define the vector $\vb$ appropriately for each problem, as follows:
\emph{(i)} Activity maximization: $\gb(\lambdab^{(0)})= \boldsymbol{\Psi}(t)^\top \vb,$ where $\vb$ is defined such that $v_j=1$ if $\alpha_j > \mu_j$, and  $v_j=0$, otherwise.
\emph{(ii)} Minimax activity shaping:
$
\gb(\lambdab^{(0)})  = \boldsymbol{\Psi}(t)^\top \eb,
$
where $\eb$ is defined such that $e_j=1$ if $\mu_j = \mu_{min}$, and $e_j=0$, otherwise.
\emph{(iii)} Least-squares activity shaping:
$
\gb(\lambdab^{(0)}) = 2\boldsymbol{\Psi}(t)^\top\boldsymbol{B}^\top\rbr{\boldsymbol{B} \boldsymbol{\Psi}(t) \boldsymbol{\lambda}^{(0)} -\vb}.
$
\emph{(iv)} Activity homogenization:
$
\gb(\lambdab^{(0)}) =  \boldsymbol{\Psi}(t)^{\top} \ln{(\boldsymbol{\Psi}(t) \boldsymbol{\lambda}^{(0)})} + \boldsymbol{\Psi}(t)^{\top} \boldsymbol{1},
$
where $\ln(\cdot)$ on a vector is the element-wise natural logarithm.
Since the activity maximization and the minimax activity shaping tasks require only one evaluation of $\Psib{(t)}$ times a vector, Algorithm~\ref{exp-algorithm} can be used directly. 
However, computing the gradient for least-squares activity shaping and activity homogenization is slightly more involved and it requires to be careful with the
order in which we perform the operations. 
Algorithm~\ref{lsq-algorithm} includes the efficient procedure to compute the gradient in the least-squares activity shaping task. Since $\Bb$ is usually sparse, it includes two multiplications 
of a sparse matrix and a vector, two matrix exponentials multiplied by a vector, and two sparse linear systems of equations.
Algorithm~\ref{hom-algorithm} summarizes the steps for efficient computation of the gradient in the activity homogenization task. Assuming again a sparse $\Bb$, it consists of two multiplication 
of a matrix exponential and a vector and two sparse linear systems of equations.

\begin{algorithm}[t]
\caption{Gradient For Least-squares Activity Shaping}
\label{lsq-algorithm}
\SetKwInOut{Input}{input}
\SetKwInOut{Output}{output}
\SetKw{KwRet}{return}
\SetKwComment{Comment}{}{}
\Input{$\Ab$, $\omega$, $t$, $\vb$, $\lambdab^{(0)}$}
\Output{$\gb(\lambdab^{(0)})$}
$\vb_1 = \Psib(t) \lambdab^{(0)}$	\Comment*[l]{\hspace{1cm}//Application of algorithm \ref{exp-algorithm}}
$\vb_2 = \Bb \vb_1$	\Comment*[l]{\hspace{1cm}//Sparse matrix vector product}
$\vb_3 = \Bb^{\top} (\vb_2-\vb)$ \Comment*[l]{\hspace{1cm}//Sparse matrix vector product}
$\vb_4 =  \Psib(t) \vb_3$ \Comment*[l]{\hspace{1cm}//Application of algorithm \ref{exp-algorithm}}
\KwRet{$2\vb_4$}
\end{algorithm}

\begin{algorithm}[t]
\caption{Gradient For Activity Homogenization}
\label{hom-algorithm}
\SetKwInOut{Input}{input}
\SetKwInOut{Output}{output}
\SetKw{KwRet}{return}
\SetKwComment{Comment}{}{}
\Input{$\Ab$, $\omega$, $t$, $\vb$, $\lambdab^{(0)}$}
\Output{$\gb(\lambdab^{(0)})$}
$\vb_1 = \Psib(t) \lambdab^{(0)}$	\Comment*[l]{\hspace{1cm}//Application of algorithm \ref{exp-algorithm}}
$\vb_2 = \ln(\vb_1)$; \\
$\vb_3 =  \Psib(t)^{\top} \vb_2$ \Comment*[l]{\hspace{1cm}//Application of algorithm \ref{exp-algorithm}}
$\vb_4 = \Psib(t)^{\top} \mathbf{1}$	\Comment*[l]{\hspace{1cm}//Application of algorithm \ref{exp-algorithm}}
\KwRet{$\vb_3 + \vb_4$}
\end{algorithm}

Equipped with an efficient way to compute of gradients, we solve the corresponding convex optimization problem for each activity shaping problem by applying the projected gradient 
descent~\cite{BoyVan04} optimization framework with the appropriate gradient\footnote{\label{note1}For nondifferential objectives, subgradient algorithms can be used instead.}.
Algorithm \ref{algorithm1} summarizes the key steps of the algorithm.

\begin{algorithm}[t]
\caption{Projected Gradient Descent for Activity Shaping}
\label{algorithm1}
Initialize $\lambdab^{(0)}$\;
\Repeat{convergence}
{
   1- Project $\lambdab^{(0)}$ into the linear space $\lambdab^{(0)}\geqslant 0$, $\boldsymbol{c^\top\lambda^{(0)}}\leqslant\boldsymbol{C}$\;
   2- Evaluate the gradient $\gb(\lambdab^{(0)})$ at $\boldsymbol{\lambda}^{(0)}$\;
   3- Update $\lambdab^{(0)}$ using the gradient $\gb(\lambdab^{(0)})$\;
}
\end{algorithm}


\section{Experimental Evaluation} 
\label{sec:evaluation}

We evaluate our activity shaping framework using both simulated and real world held-out data, and show that our approach significantly outperforms several baselines. 

\subsection{Experimental Setup}
Here, we briefly present our data, evaluation schemas, and settings.

{\bf Dataset description and network inference.}
We use data gathered from Twitter as reported in~\cite{MeeHadBenGum10}, which comprises of all public tweets posted by 60,000 users during a 8-month period, from January 2009
to September 2009.
For every user, we record the times she uses any of the following six url shortening services: Bitly , TinyURL, Isgd, TwURL, SnURL, Doiop (refer to Appendix~\ref{append:evaluation} for details).
We evaluate the performance of our framework on a subset of 2,241 active users, linked by 4,901 edges, which we call 2K dataset, and we evaluate
its scalability on the overall 60,000 users, linked by $\sim$ 200,000 edges, which we call 60K dataset.
The 2K dataset accounts for 691,020 url shortened service uses while the 60K dataset accounts for $\sim$7.5 million uses.
Finally, we treat each service as independent cascades of events.

In the experiments, we estimated the nonnegative influence matrix $\Ab$ and the exogenous intensity $\boldsymbol{\lambda}^{(0)}$ using maximum log-likelihood, as in previous
work~\cite{ZhoZhaSon13, ZhoZhaSon13b, ValGomGum14}.
We used a temporal resolution of one minute and selected the bandwidth $\omega=0.1$ by cross validation. 
Loosely speaking, $\omega=0.1$ corresponds to loosing 70\% of the initial influence after 10 minutes,
which may be explained by the rapid rate at which each user'{} news feed gets updated. 

{\bf Evaluation schemes.} We focus on three tasks: capped activity maximization, minimax activity shaping, and least square activity shaping. We set the total budget to $C=0.5$,
which corresponds to supporting a total extra activity equal to $0.5$ actions per unit time, and assume all users entail the same cost.
In the capped activity maximization, we set the upper limit of each user'{}s intensity, $\boldsymbol{\alpha}$, by adding a nonnegative random vector to their inferred initial intensity.
In the least-squares activity shaping, we set $\boldsymbol{B} = \bold{I}$ and aim to create three groups of users, namely less-active, moderate, and super-active users.
%
%
We use three different evaluation schemes, with an increasing resemblance to a real world scenario:

\begin{figure*}[!t]
  \centering
  \setlength{\tabcolsep}{6pt}
  \begin{tabular}{ccc}
          \includegraphics[width=0.30\textwidth]{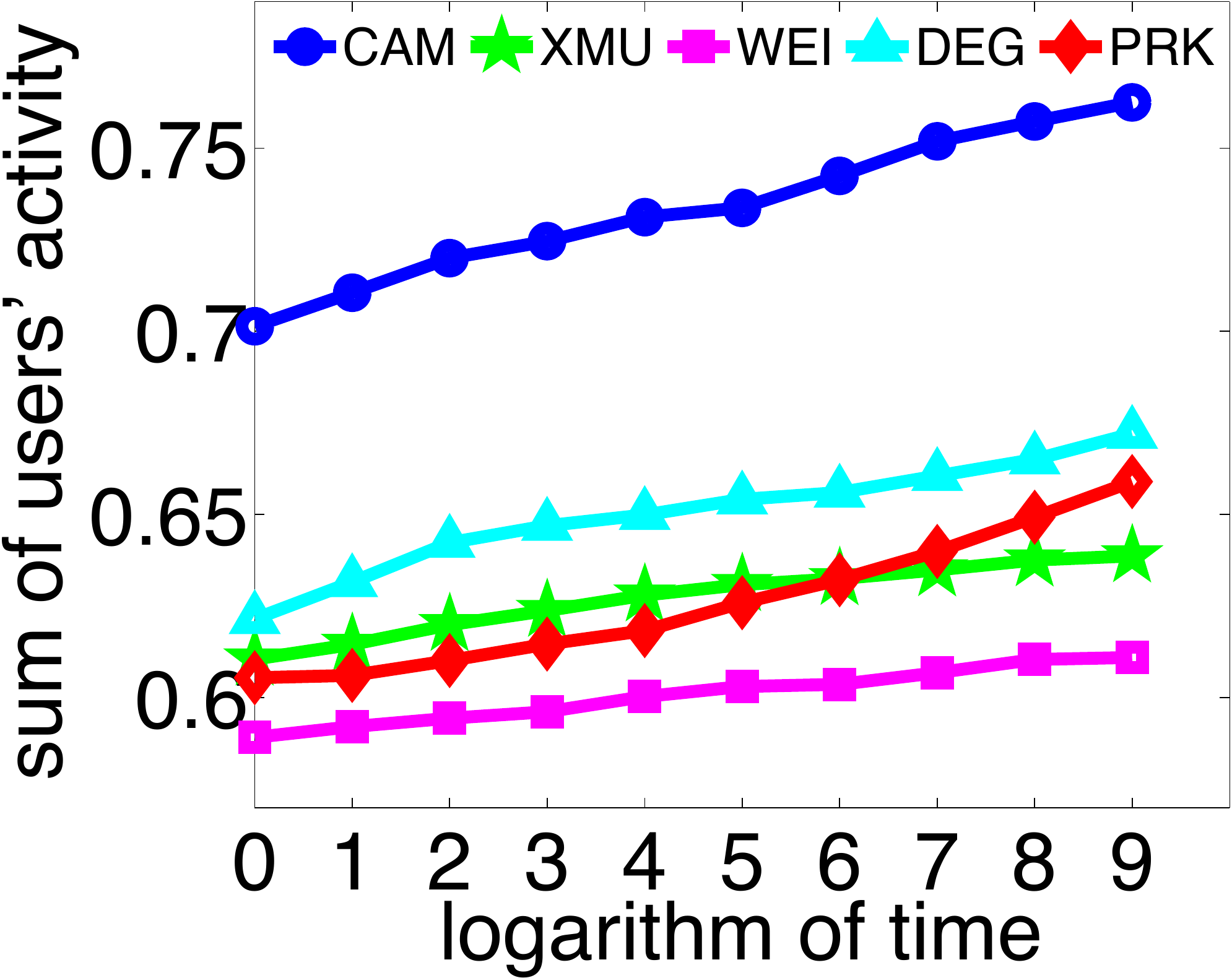} &
          \includegraphics[width=0.30\textwidth]{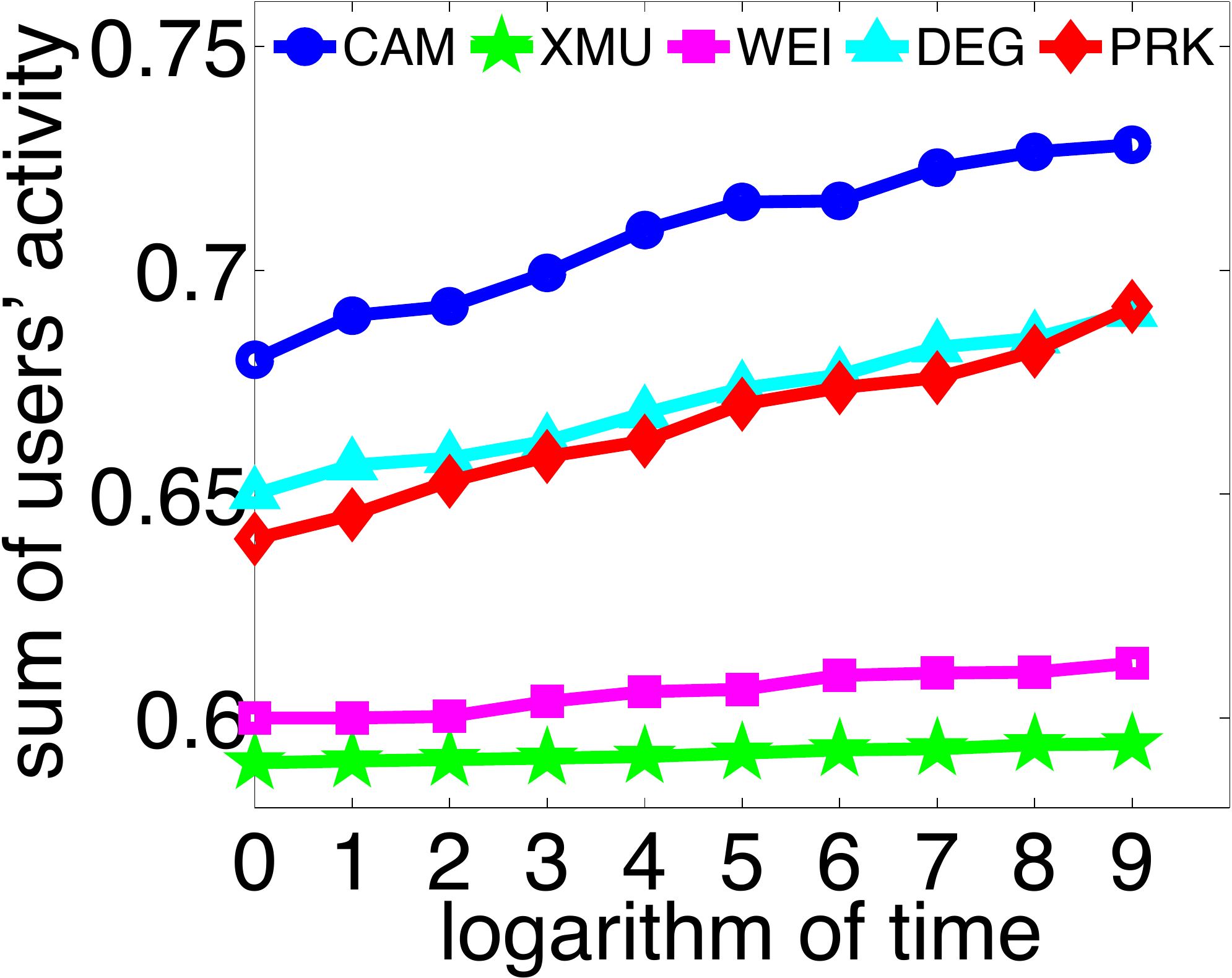} &
          \includegraphics[width=0.33\textwidth]{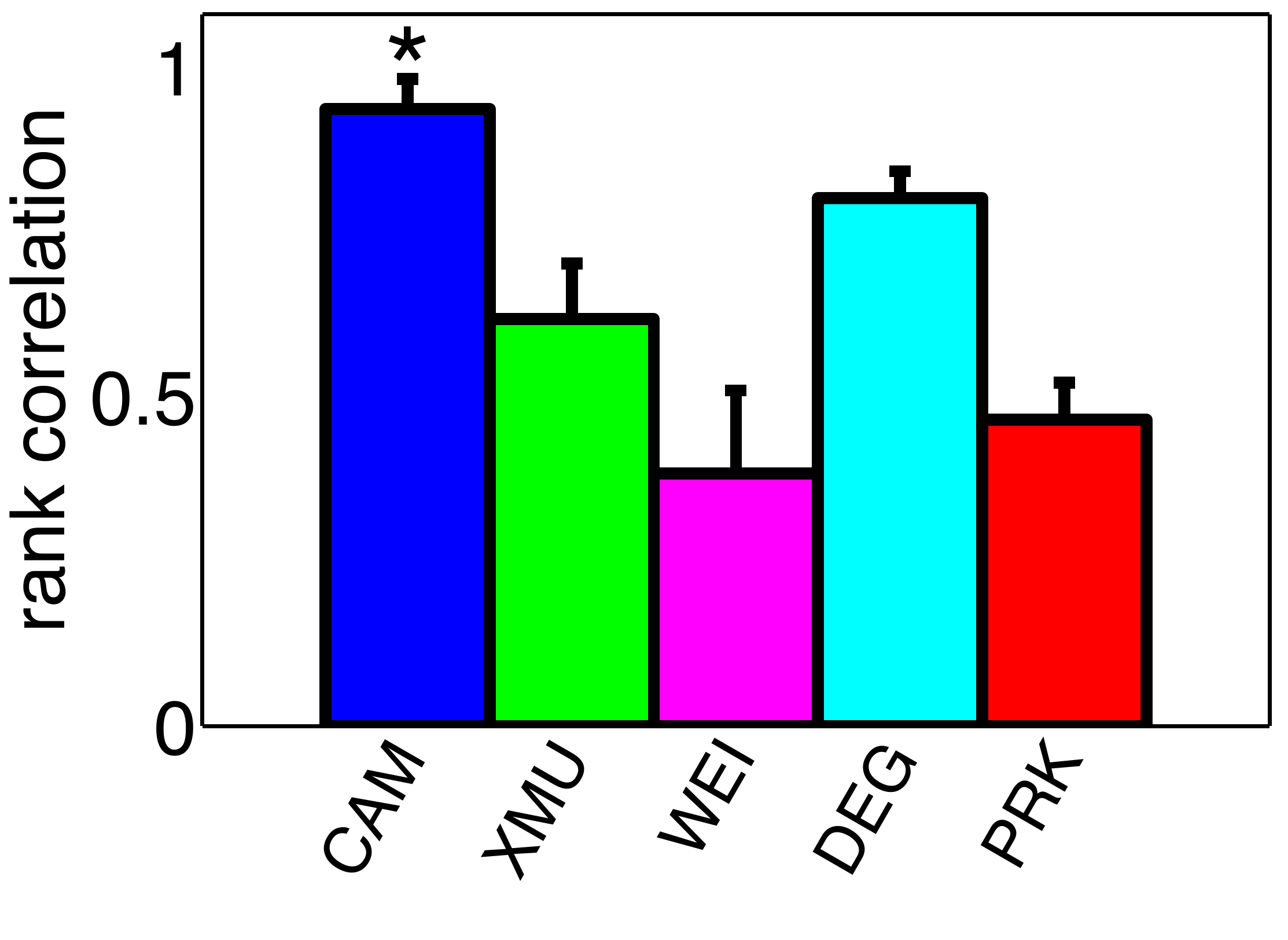} \\
          \includegraphics[width=0.30\textwidth]{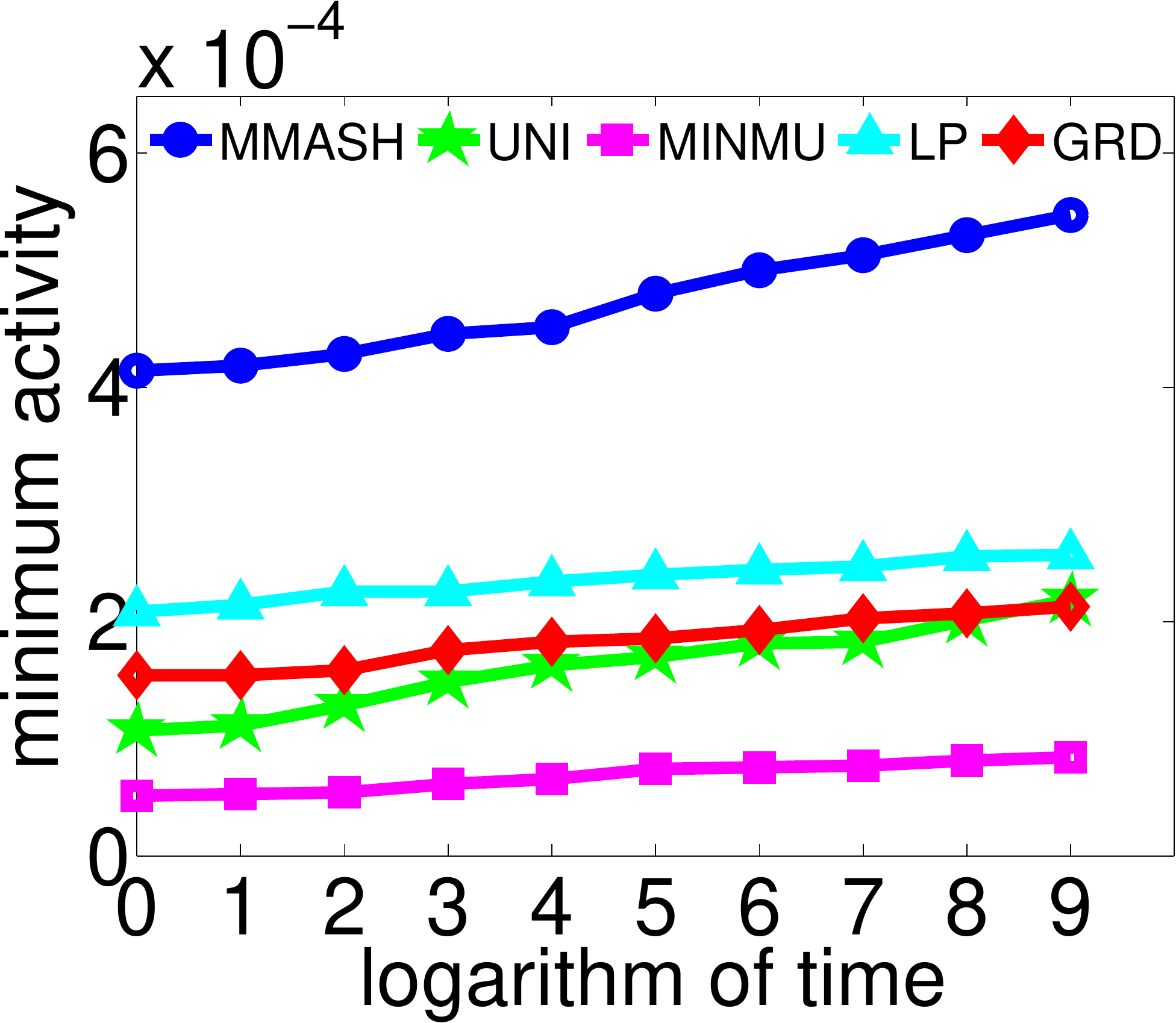} &
          \includegraphics[width=0.30\textwidth]{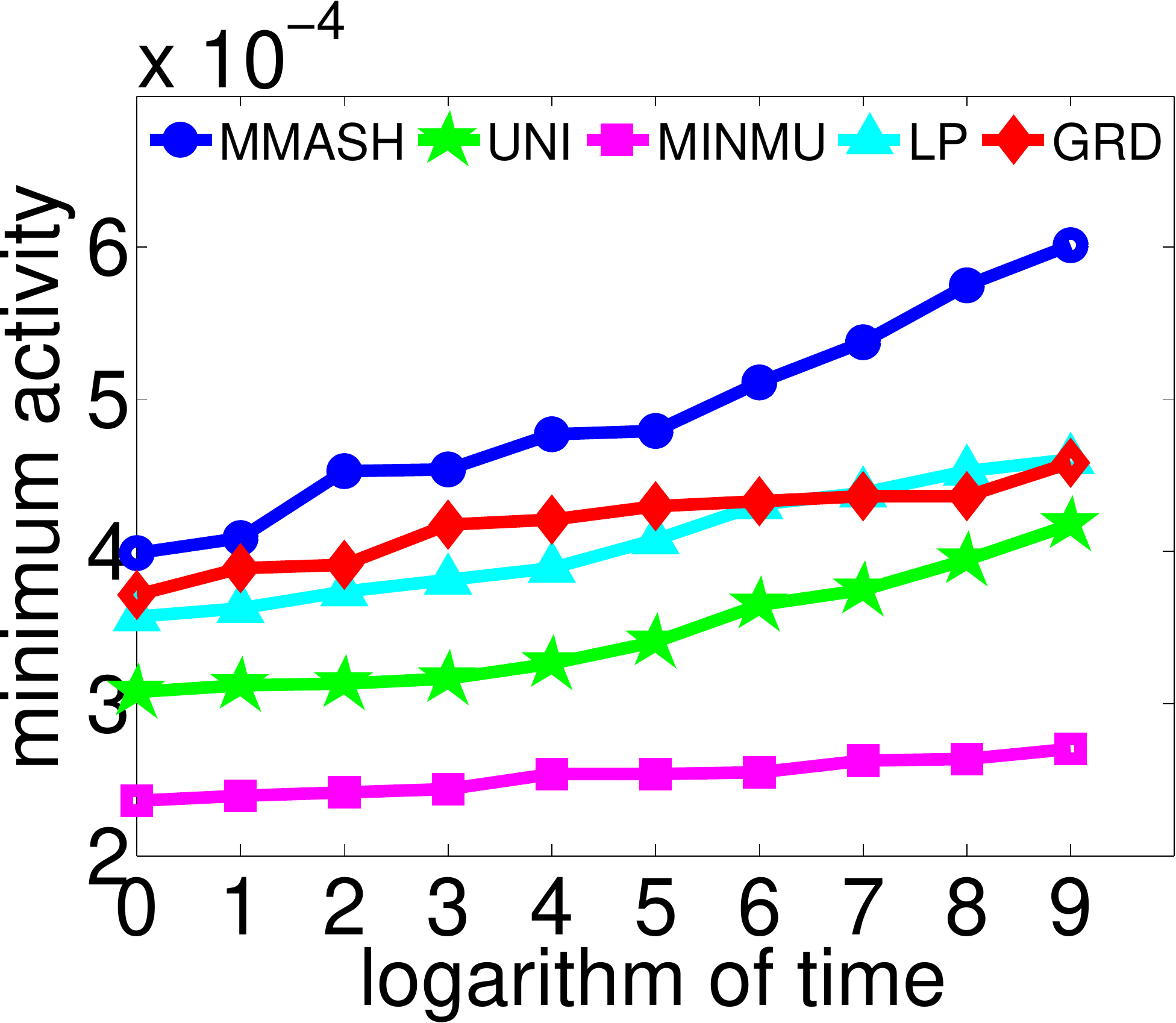} &
          \includegraphics[width=0.33\textwidth]{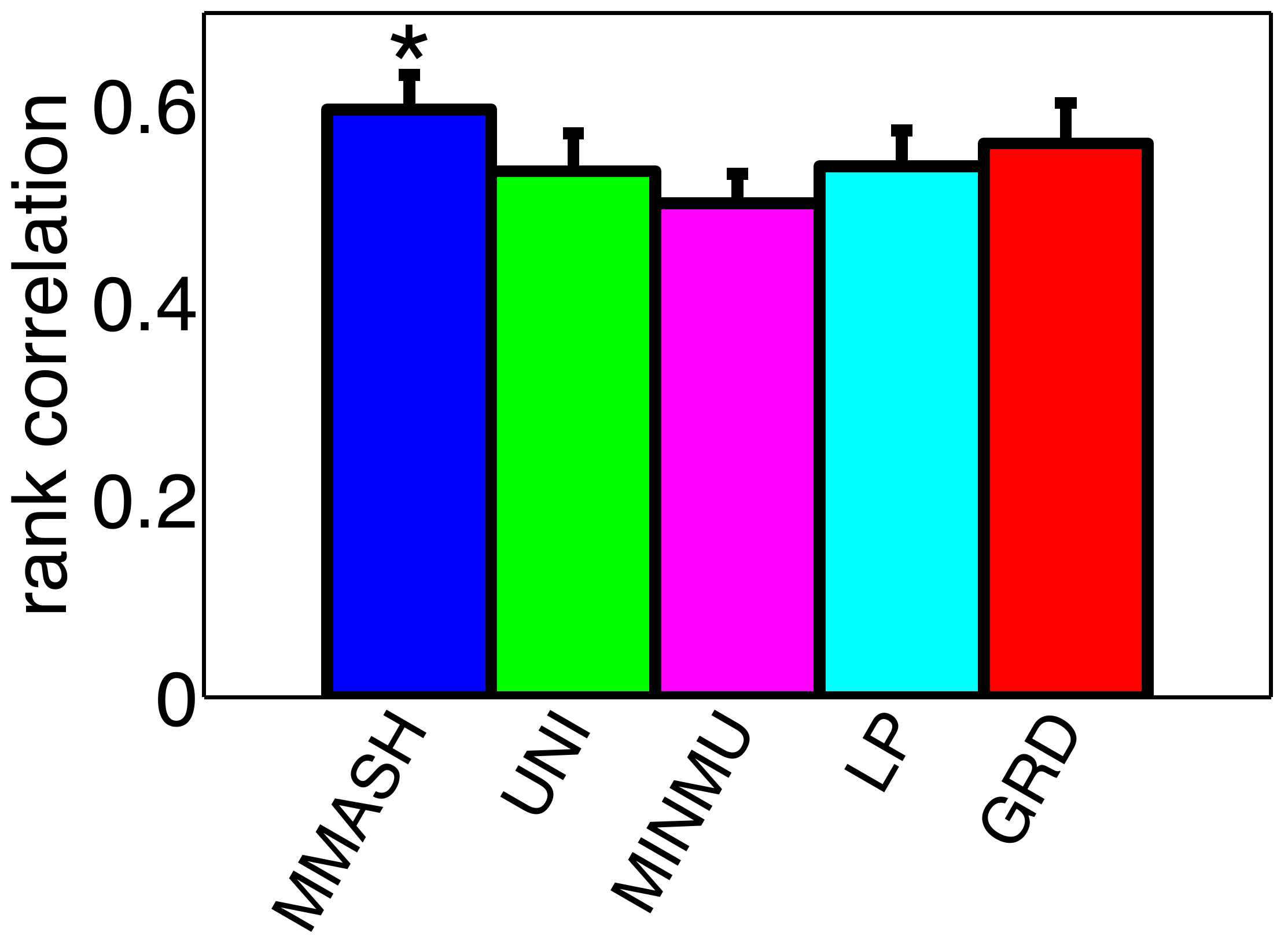} \\
          \includegraphics[width=0.30\textwidth]{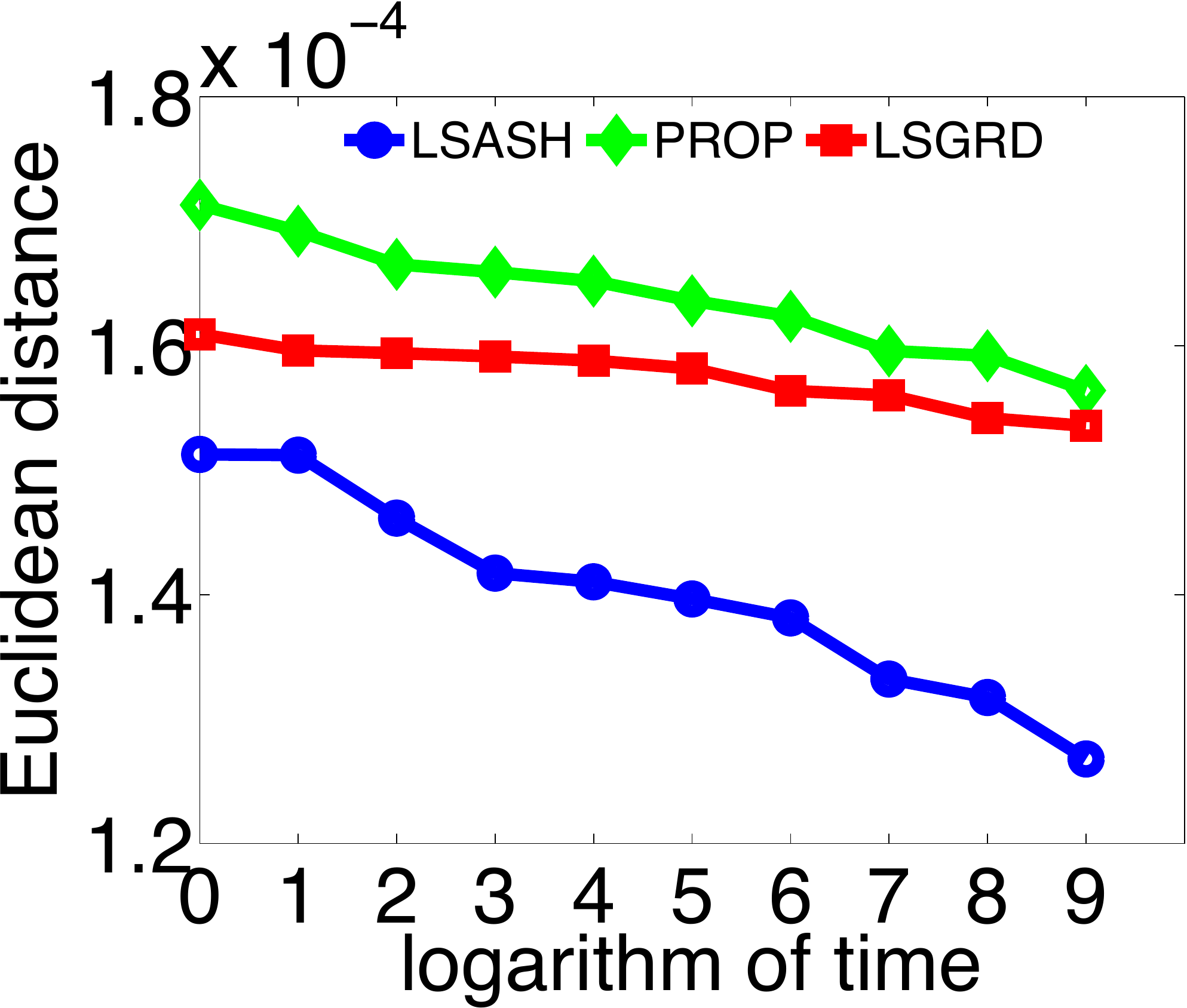} &
          \includegraphics[width=0.30\textwidth]{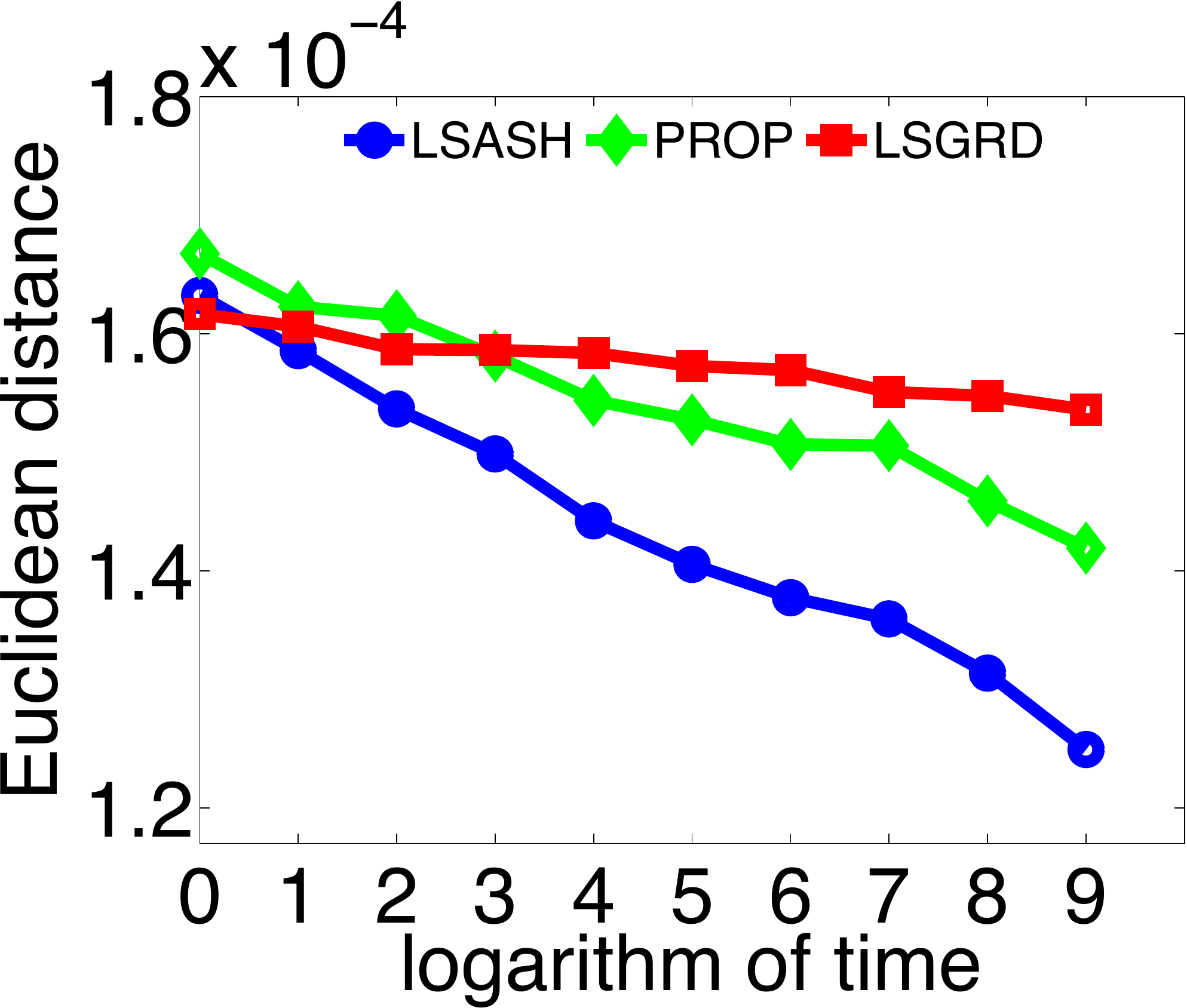} &
          \includegraphics[width=0.33\textwidth]{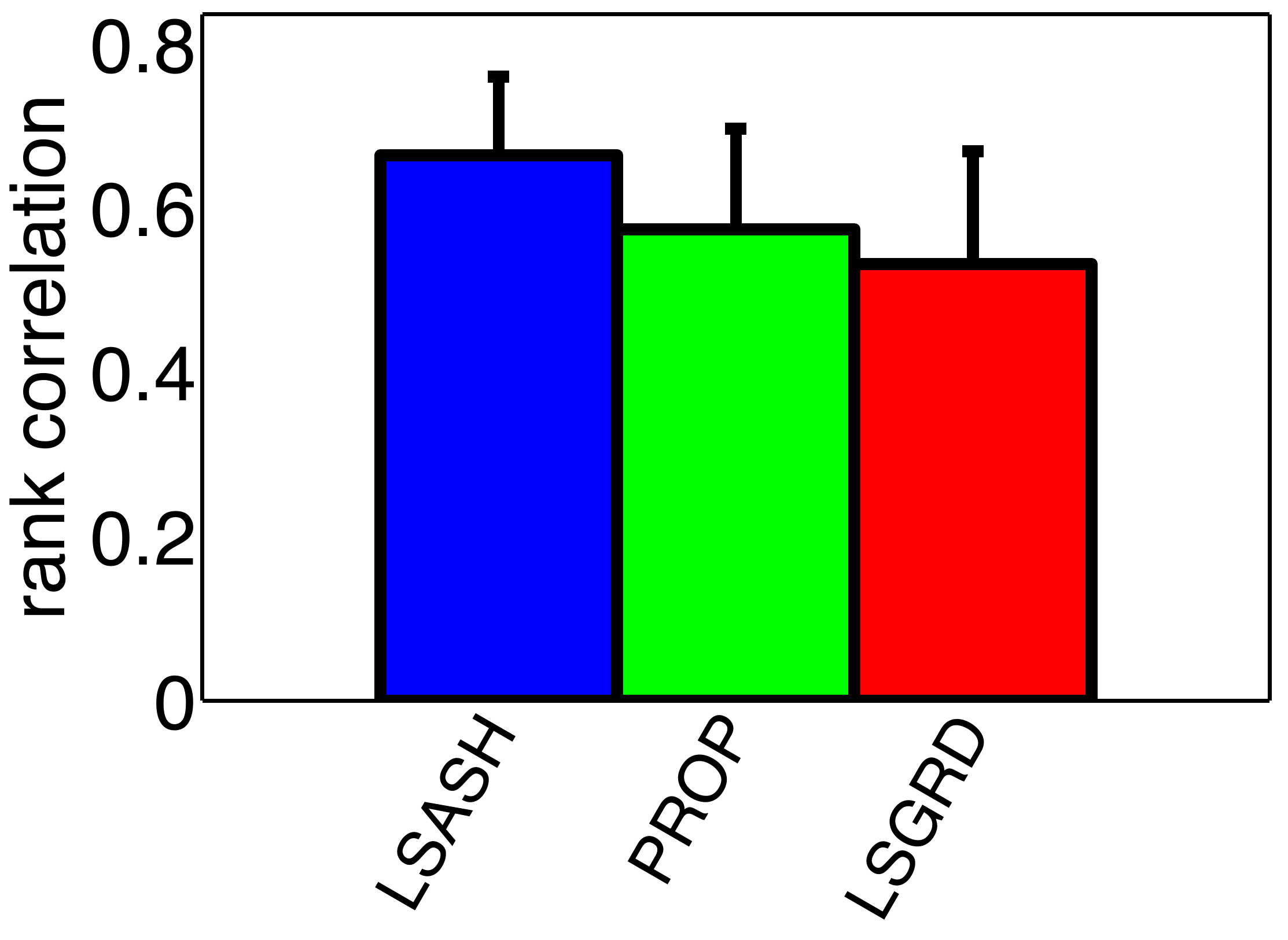} \\
          (a) Theoretical objective & (b) Simulated objective & (c) Held-out data
  \end{tabular}
  \caption{Row 1: Capped activity maximization. Row 2: Minimax activity shaping. Row 3: Least-squares activity shaping. * means statistical significant at level of 0.01 with paired t-test between our method and the second best}
  \label{fig:results}
\end{figure*}

\emph{Theoretical objective}: We compute the expected overall (theoretical) intensity by applying Theo\-rem~\ref{theo:lin_rel} on the optimal exogenous event intensities, $\lambdab_{opt}^{(0)}$,
to each of the three activity shaping tasks, as well as the learned $\Ab$ and $\omega$. We then compute and report the value of the objective functions. 

\emph{Simulated objective}: We simulate $50$ cascades with Ogata'{}s thinning algorithm~\cite{Ogata1981}, using the optimal exogenous event intensities, $\lambdab_{opt}^{(0)}$, to each
of the three activity shaping tasks, and the learned $\Ab$ and $\omega$.
We then estimate empirically the overall event intensity based on the simulated cascades, by computing a running average over non-overlapping time windows, 
and report the value of the objective functions based on this estimated overall intensity.
%
%
Appendix~\ref{append:temp} provides a comparison between the simulated and the theoretical objective.

\emph{Held-out data}: The most interesting evaluation scheme would entail carrying out real interventions in a social platform. However, since this is very challenging to do,
instead, in this evaluation scheme, we use held-out data to simulate such process, proceeding as follows.
We first partition the $8$-month data into $50$ five-day long contiguous intervals. Then, we use one interval for training and the remaining $49$ intervals for testing.
Suppose interval $1$ is used for training, the procedure is as follows:\\[-4mm]
\begin{enumerate}[noitemsep, nolistsep]
  \item We estimate $\Ab_1$, $\omega_1$ and $\lambdab_1^{(0)}$ using the events from interval $1$. Then, we fix $\Ab_1$ and $\omega_1$,  and estimate $\lambdab_i^{(0)}$ for all other intervals,
  $i=2,\ldots,49$.
  \item Given $\Ab_1$ and $\omega_1$, we find the optimal exogenous event intensities, $\lambdab_{opt}^{(0)}$, for each of the three activity shaping task, by solving the associated convex program.
  We then sort the estimated $\lambdab_i^{(0)}$ ($i=2,\ldots,49$) according to their similarity to $\lambdab_{opt}^{(0)}$, using the Euclidean distance $\|\lambdab_{opt}^{(0)} - \lambdab_i^{(0)}\|_2$.
  \item We estimate the overall event intensity for each of the $49$ intervals ($i=2,\ldots,49$), as in the ``simulated objective'' evaluation scheme, and sort these intervals according to the value of their
  corresponding objective function.
  \item Last, we compute and report the rank correlation score between the two orderings obtained in step 2 and 3.\footnote{rank correlation = number of pairs with consistent ordering / total number of
  pairs.} The larger the rank correlation, the better the method.
\end{enumerate}
We repeat this procedure 50 times, choosing each different interval for training once, and compute and report the average rank correlations. 

It is beneficial to emphasize that the held-out experiments are essentially evaluating prediction performance on test sets. For instance, suppose we are given a diffusion network and two different configuration of incentives. We will shortly show our method can predict more accurately which one will reach the activity shaping goal better. This means, in turn, that if we incentivize the users according to our method's suggestion, we will achieve the target activity better than other heuristics.

Alternatively, one can understand our evaluation scheme like this: if one applies the incentive (or intervention) levels prescribed by a method, how well the predicted outcome coincides with the reality in the test set? A good method should behavior like this: the closer the prescribed incentive (or intervention) levels to the estimated base intensities in test data, the closer the prediction based on training data to the activity level in the test data. In our experiment, the closeness in incentive level is measured by the Euclidean distance, the closeness between prediction and reality is measured by rank correlation.


\subsection{Activity Shaping Results}
In this section, the results for three activity shaping tasks evaluated on the three schemas are presented.

{\bf Capped activity maximization (CAM).} We compare to a number of alternatives. XMU: heuristic based on $\mub(t)$ without optimization; DEG and WEI: heuristics based on the degree of the user;
PRANK: heuristic based on page rank (refer to Appendix~\ref{append:evaluation} for further details).
The first row of Figure~\ref{fig:results} summarizes the results for the three different evaluation schemes. We find that our method (CAM) consistently outperforms the alternatives.
%
For the theoretical objective, CAM is 11 \% better than the second best, DEG. The difference in overall users'{} intensity from DEG is about $0.8$ which, roughly speaking, leads to at least an increase of
about $0.8 \times 60 \times 24 \times 30 = 34,560$ in the overall number of events in a month.
In terms of simulated objective and held-out data, the results are similar and provide empirical evidence that, compared to other heuristics, degree is an appropriate surrogate for influence, while,
based on the poor performance of XMU, it seems that high activity does not necessarily entail being influential.
To elaborate on the interpretability of the real-world experiment on held-out data, consider for example the difference in rank correlation between CAM and DEG, which is almost $0.1$.
Then, roughly speaking, this means that incentivizing users based on our approach accommodates with the ordering of real activity patterns in $0.1\times \frac{ 50 \times 49 }{2} = 122.5$ more pairs of realizations.

{\bf Minimax activity shaping (MMASH).} We compare to a number of alternatives. UNI: heuristic based on equal allocation; MINMU: heuristic based on $\mub(t)$ without optimization; LP: linear programming based heuristic; GRD: a greedy approach to leverage the activity (see Appendix~\ref{append:evaluation} for more details).
The second row of Figure~\ref{fig:results} summarizes the results for the three different evaluation schemes. We find that our method (MMASH) consistently outperforms the alternatives.
For the theoretical objective, it is about 2$\times$ better than the second best, LP. Importantly, the difference between MMASH and LP is not trifling and the least active user carries out
$2 \times 10^{-4} \times 60 \times 24 \times 30=4.3$ more actions in average over a month.
As one may have expected, GRD and LP are the best among the heuristics.
The poor performance of MINMU, which is directly related to the objective of MMASH, may be because it assigns the budget to a low active user, regardless of
their influence.
However, our method, by cleverly distributing the budget to the users whom actions trigger many other users'{} actions (like those ones with low activity), it benefits from the budget most.
In terms of simulated objective and held-out data, the algorithms'{} performance become more similar.

{\bf Least-squares activity shaping (LSASH).} We compare to two alternatives. PROP: Assigning the budget proportionally to the desired activity; LSGRD: greedily allocating budget according the
difference between current and desired activity (refer to Appendix~\ref{append:evaluation} for more details).
The third row of Figure~\ref{fig:results} summarizes the results for the three different evaluation schemes. We find that our method (LSASH) consistently outperforms the alternatives.
Perhaps surprisingly, PROP, despite its simplicity, seems to perform slightly better than LSGRD.
This is may be due to the way it allocates the budget to users, \eg, it does not aim to strictly fulfill users'{} target activity but benefit more users by assigning budget proportionally.
Refer to Appendix~\ref{append:visual} for additional experiments.

In all three tasks, longer times lead to larger differences between our method and the alternatives. This occurs because the longer the time, the more endogenous activity is triggered by network
influence, and thus our framework, which models both endogenous and exogenous events, becomes more suitable.

\subsection{Sparsity and Activity Shaping}
\label{sparsity}
In some applications there is a limitation on the number of users we can incentivize.
In our proposed framework, we can handle this requirement by including a sparsity constraint on the optimization problem.
In order to maintain the convexity of the optimization problem, we consider a $l_1$ regularization term, where a regularization parameter $\gamma$ provides the trade-off between sparsity and
the activity shaping goal:
\begin{equation}
	\label{eq:generalized-activity-maximization-sparsity}
	\begin{array}{ll}
		\mbox{maximize}_{\mub(t),\lambdab^{(0)}} & U( \mub(t)) - \gamma || \boldsymbol{\lambda}^{(0)} ||_1 \\
		\mbox{subject to} & \mub(t)=\Psib(t) \lambdab^{(0)},\quad \cbb^\top \lambdab^{(0)} \leqslant C,\quad \lambdab^{(0)} \geqslant 0
	\end{array}
\end{equation}

Tables \ref{tb:cons_ave} and  \ref{tb:minimax} demonstrate the  effect of different values of regularization parameter on \emph{capped activity maximization} and \emph{minimax activity shaping}, respectively. When $\gamma$ is small, the minimum intensity is very high. On the contrary, large values of $\gamma$ imposes large penalties on the number of non-zero intensities which results in a sparse and applicable manipulation. Furthermore, this may avoid using all the budget. When dealing with unfamiliar application domains, cross validation may help to find an appropriate trade-off between sparsity and objective function.

\begin{table} [h]
\small
\centering
\begin{tabular}{|c|c|c|c|}
 \hline
$\gamma$ & $\#$ Non-zeros  & Budget consumed  & Sum of activities \\ \hline \hline
0.5 & 2101 & 0.5 & 0.69 \\ \hline
0.6    & 1896 & 0.46 & 0.65 \\ \hline
0.7  & 1595 & 0.39 & 0.62 \\ \hline
0.8  & 951 & 0.21 & 0.58 \\ \hline
0.9  & 410 & 0.18 & 0.55 \\ \hline
1.0  & 137 & 0.13 & 0.54 \\ \hline
\end{tabular}
\caption{Sparsity properties of  capped activity maximization.}
\label{tb:cons_ave}
\end{table}

\begin{table}[h]
\small
\centering
\begin{tabular}{|c|c|c|c|}
 \hline
$\gamma(\times 10^{-3})$ & $\#$ Non-zeros  & Budget Consumed& $u_{min}(\times 10^{-3}) $ \\ \hline \hline
0.6 & 1941 & 0.49 & 0.38 \\ \hline
0.7    & 881 & 0.17 & 0.22 \\ \hline
0.8  & 783 & 0.15 & 0.21 \\ \hline
0.9  & 349 & 0.09 & 0.16 \\ \hline
1.0  & 139 & 0.06 & 0.12 \\ \hline
1.1  & 102 & 0.04 & 0.11 \\ \hline
\end{tabular}
\caption{Sparsity properties of  minimax activity shaping.}
\label{tb:minimax}
\end{table}

\subsection{Scalability}
\label{scale}

The most computationally demanding part of the proposed algorithm is the evaluation of matrix exponentials,
which we scale up by utilizing techniques from matrix algebra, such as GMRES and Al-Mohy methods. As a result, we are able to run our methods in a reasonable amount of
time on the 60K dataset, specifically, in comparison with a naive implementation of matrix exponential evaluations.
The naive implementation of the algorithm requires computing the matrix exponential once, and using it in (non-sparse huge) matrix-vector multiplications, \ie,
$$
T_{naive} = T_{\boldsymbol{\Psi}} + k  T_{prod}.
$$
Here, $ T_{\boldsymbol{\Psi}}$ is the time to compute $\boldsymbol{\Psi}(t)$, which itself comprised of three parts; matrix exponential computation, matrix inversion and matrix multiplications.
$T_{prod}$ is the time for multiplication between the large non-sparse matrix and a vector plus the time to compute the inversion via solving linear systems of equation.
Finally, $k$ is the number of gradient computations, or more generally, the number of iterations in any gradient-based iterative optimization.
The dominant factor in the naive approach is the matrix exponential. It is computationally demanding and practically inefficient for more than 7000 users.

In contrast, the proposed framework benefits from the fact that the gradient depends on $\boldsymbol{\Psi}(t)$ only through matrix-vector products. Thus, the running time of our activity shaping framework will be written as
$$
T_{our} = k  T_{grad},
$$
 where $T_{grad}$ is the time to compute the gradient which itself comprises  the time required to solve a couple of linear systems of equations and the time to compute a couple of exponential matrix-vector multiplication.

 \begin{figure} [t]
        \centering
        \begin{tabular}{cc}
          \includegraphics[width=0.35\textwidth]{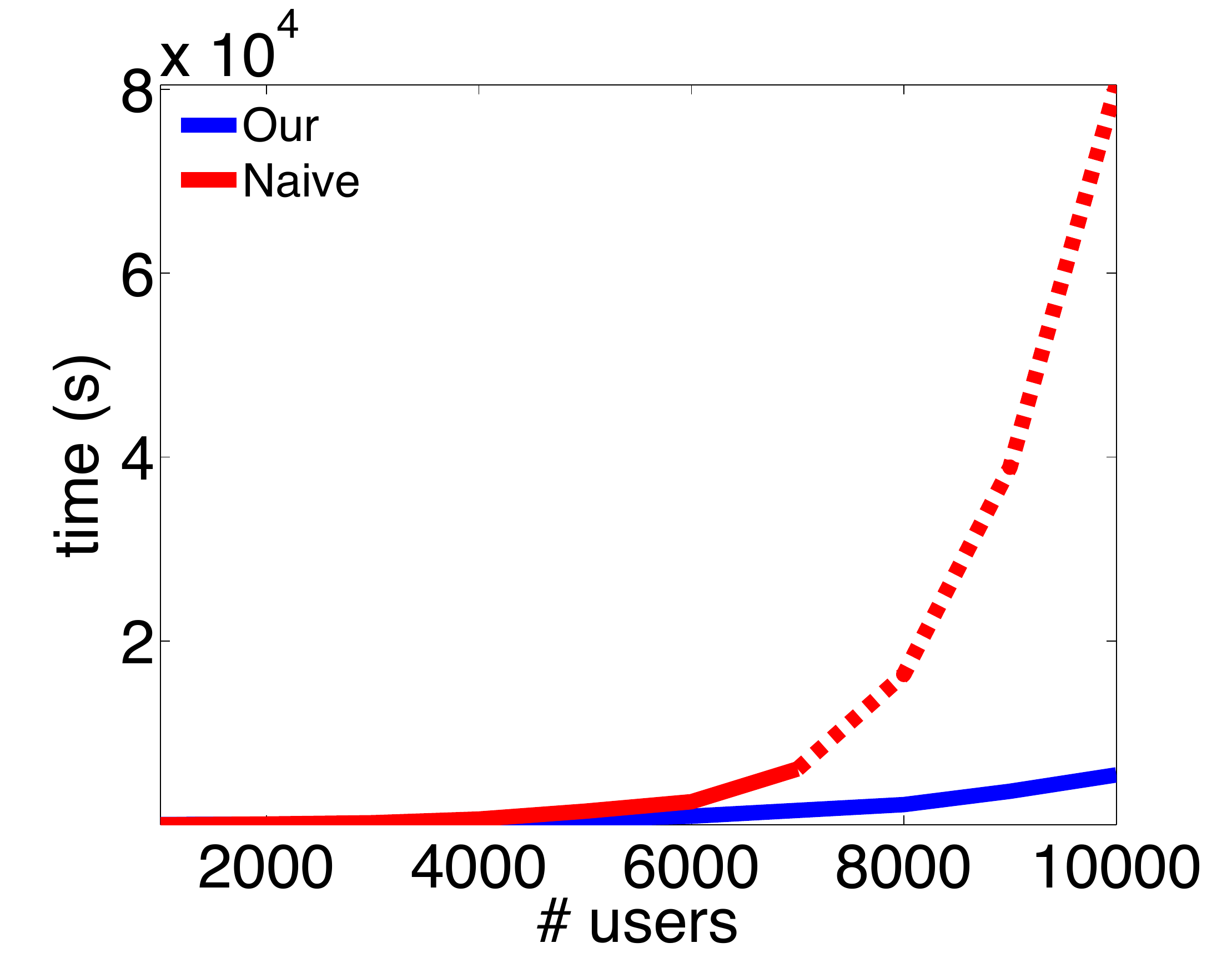}\label{fig:scale10} &
      	  \includegraphics[width=0.35\textwidth]{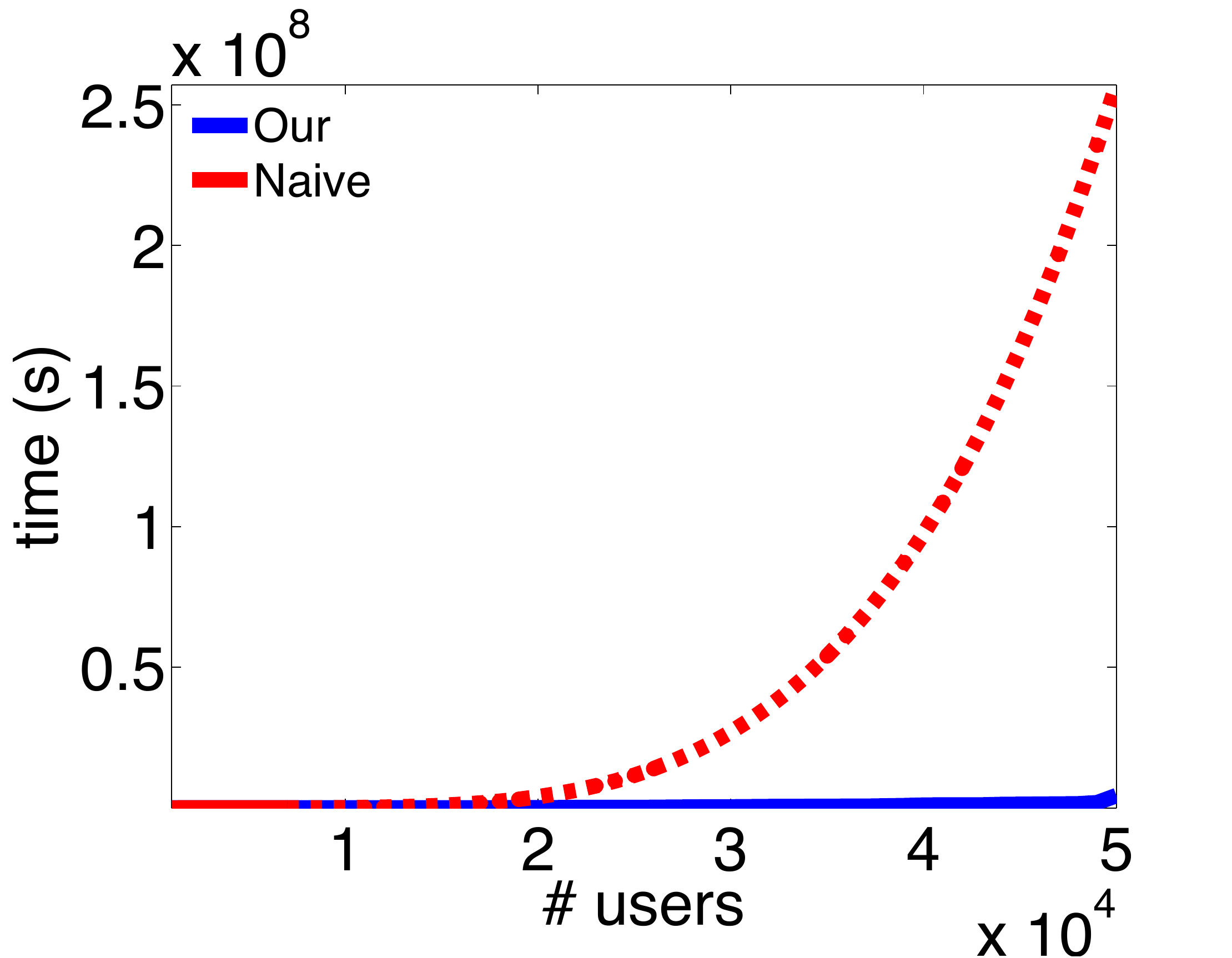}  \label{fig:scale50} \\
	  (a) For 10,000 users. &
	  (b) For 50,000 users.
	 \end{tabular}
       \caption{Scalability of least-squares activity shaping.} \label{fig:scalability}
\end{figure}

Figure \ref{fig:scalability} demonstrates $T_{our} $ and $T_{naive}$ with respect to the number of users.  For better visualization we have provided two graphs for up to 10,000 and 50,000 users, respectively.  We set $k$ equal to the number of users. Since the dominant factor in the naive computation method is matrix exponential, the choice of $k$ is not that determinant. The time for computing matrix exponential is interpolated  for more than 7000 users; and the interpolated total time, $T_{naive}$, is shown in red dashed line. These experiments are done in a machine equipped with one 2.5 GHz AMD Opteron Processor. This graph clearly shows the significance of designing an scalable algorithm.

\begin{figure*} [!t]
        \centering
          \setlength{\tabcolsep}{6pt}
        \begin{tabular}{ccc}
              \includegraphics[width=0.30\textwidth]{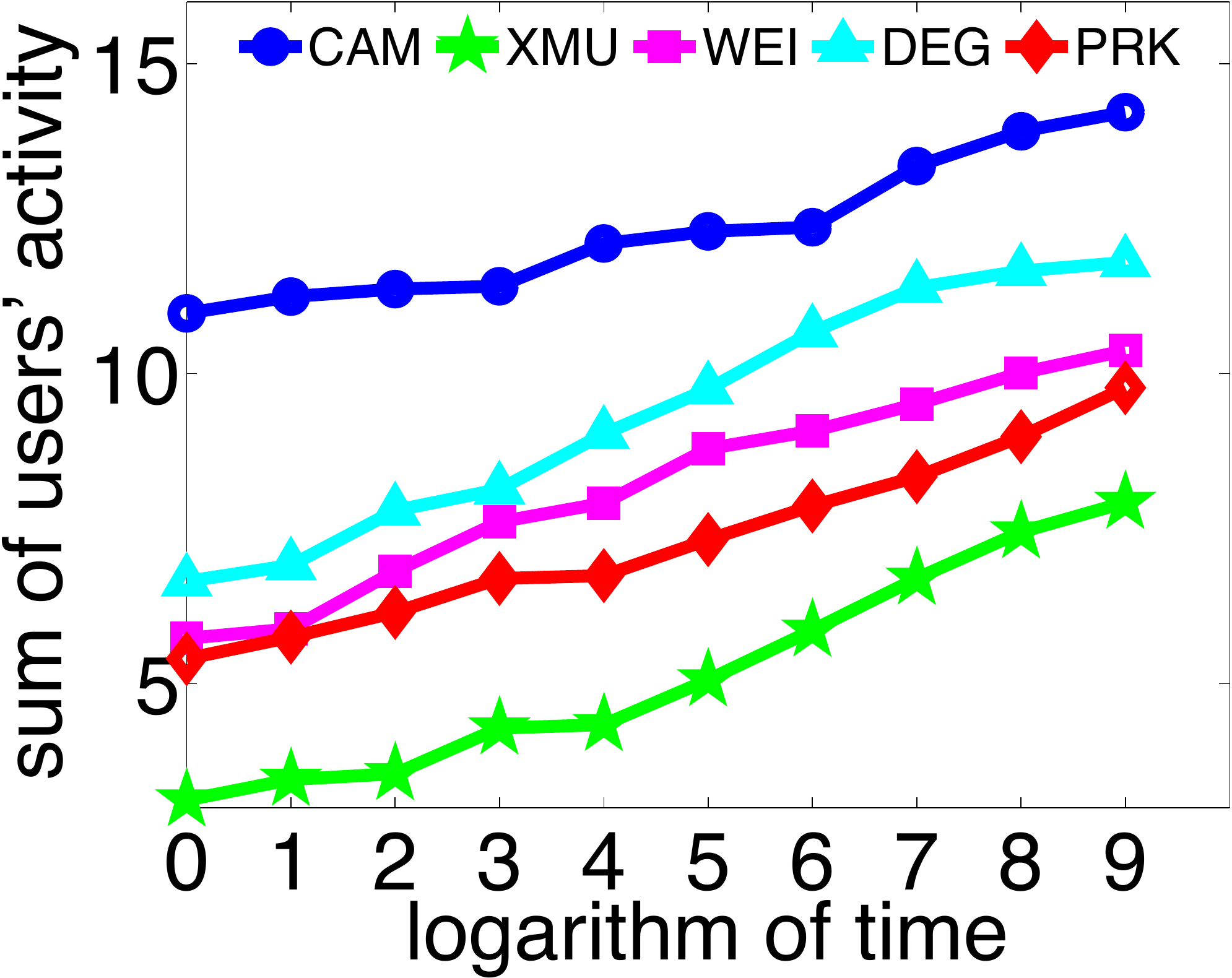} \label{fig:results:large:consave:2k} &
              \includegraphics[width=0.30\textwidth]{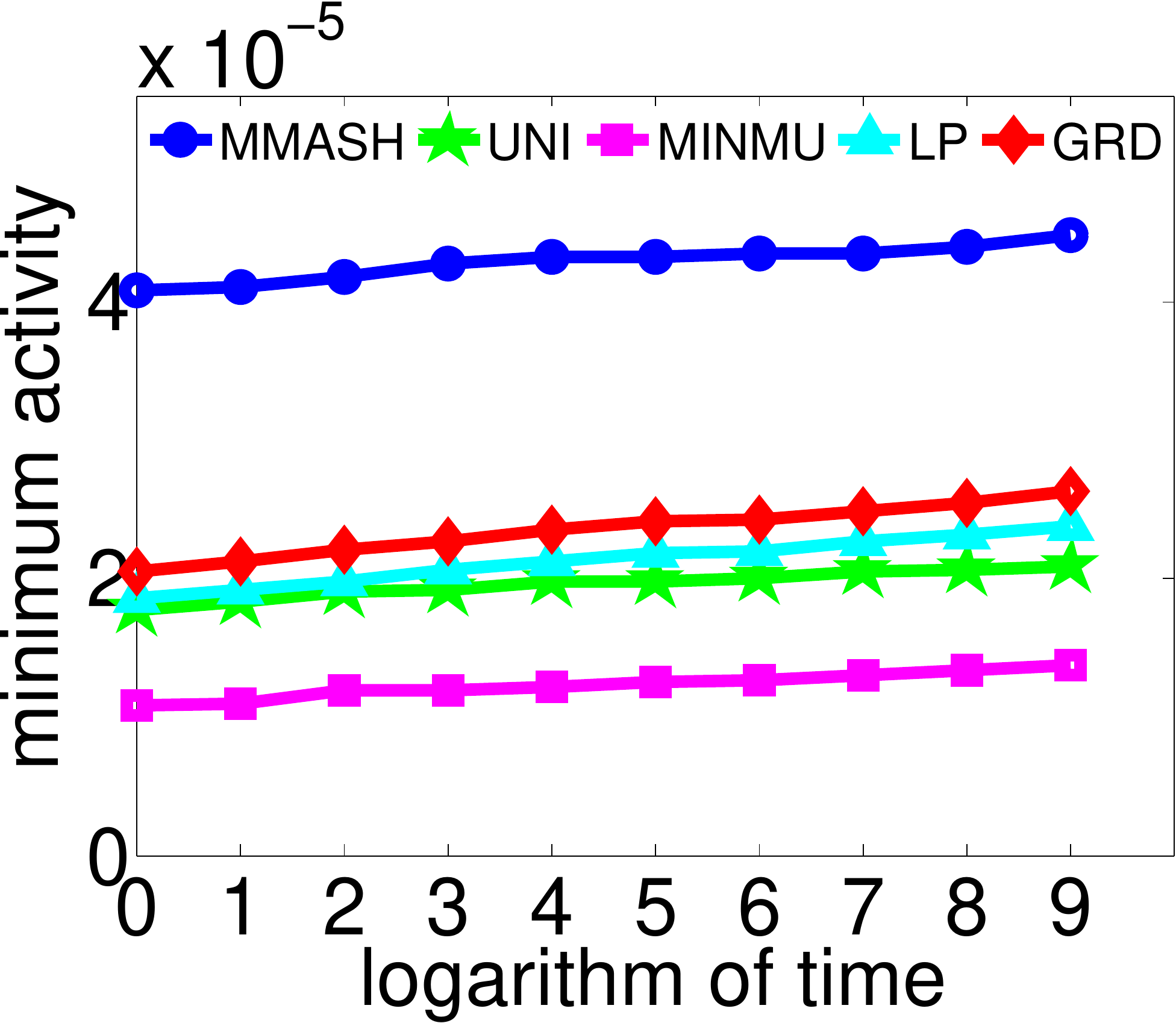}  \label{fig:results:large::minimax} &
              \includegraphics[width=0.30\textwidth]{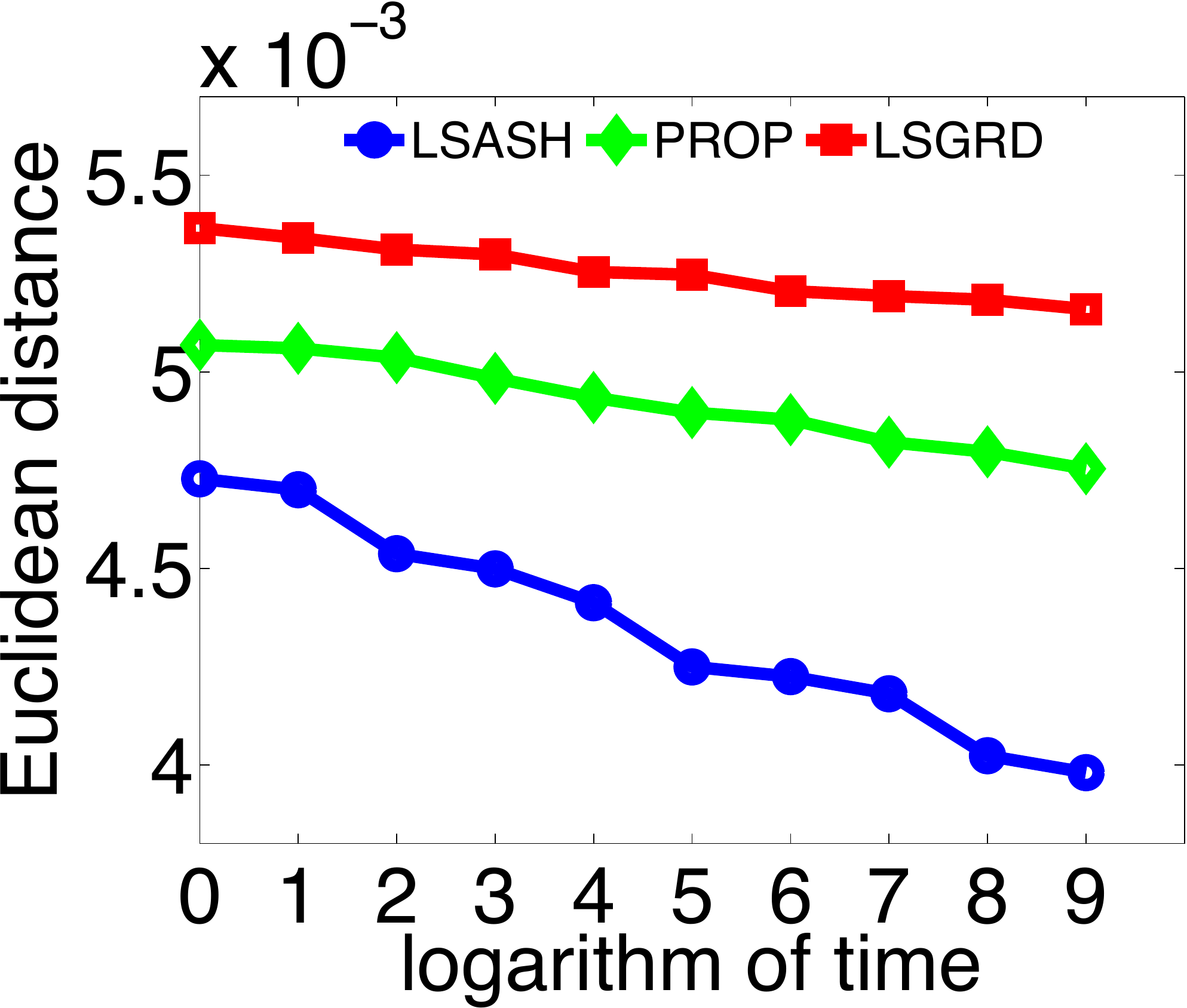}   \label{fig:results:large:lsqr} \\
                 \footnotesize
                (a) Capped activity maximization. &
                \footnotesize
                (b) Minimax activity shaping. &
                \footnotesize
                (c) Least-square activity shaping.
         \end{tabular}
         \caption{Activity shaping on the 60K dataset.}
         \label{fig:results:large}
\end{figure*}

Figure \ref{fig:results:large} shows the results of running our large-scale algorithm on the 60K dataset evaluated via theoretical objective function. We observe the same patterns as 2K dataset. Especially, the proposed method consistently outperforms the heuristic baselines. Heuristic baselines provide similar performance as for the 2K dataset. DEG shows up again as a reasonable surrogate for influence, and the poor performance of XMU on activity maximization shows that high activity does not necessarily mean being more influential. For \emph{minimax activity shaping} we observe MMASH is superior to others in $2 \times 10^{-5}$ actions per unit time, which means that the person with minimum activity uses the service $2 \times 10^{-5} \times 60 * 24 * 30 = 0.864$ times more compared to the best heuristic baseline. An increase in the activity per month of $0.864$  is not a big deal itself, however, if we consider the scale at which the network's activity is steered, we can deduce that now the service is 
guaranteeing, at least in theory, about $60000 \times 0.864 = 51840$ more adoptions monthly. As shown by the experiments on real-world held-out data, our approach for activity shaping outperforms all the considered heuristic baselines.

\section{Summary and Discussion}
In this paper, we introduced the activity shaping problem, which is a generalization of the influence maximization problem, and it allows for more elaborate goal functions. 
Our model of social activity is based on multivariate Hawkes processs, and via a connection to branching processes, we manage to derive a linear connection between the exogenous activity 
(\ie, the part that can be easily manipulated via incentives) and the overall network activity. 
This connection enables developing a convex optimization framework for activity shaping, deriving the necessary incentives to reach a global activity pattern in the network. The method is evaluated 
on both synthetic and real-world held-out data and is shown to outperform several heuristics.

We acknowledge that our method has indeed limitations.
For example, our current formulation assumes that exogenous events are constant over time. 
Thus, subsequent evolution in the point process is a mixture of endogenous and exogenous events.
However, in practice, the shaping incentives need to be doled out throughout the evolution of the process, \eg, in a sequential decision making setting. 
Perhaps surprisingly, our framework can be generalized to time-varying exogenous events, at the cost of stating some of the theoretical results in a convolution form, 
as follows:
\begin{itemize}
\item
 Lemma~\ref{lem:br_intensity} needs to be kept in convolution form, \ie, $\mub^{(k)}(t) = \Gb^{(\star k)}(t) \star  \lambdab^{(0)}(t)$.  
 The sketch of the proof is very similar, and we only need to further exploit the associativity property of the convolution at the inductive step, to prove the hypothesis holds for $k+1$:
\begin{equation}
 \mub^{(k+1)}(t) = \int_{0}^t \Gb(t-s)\, \left( \Gb^{(\star k)}(s)\,\star \, \lambdab^{(0)}(s) \right) \, ds
    = \Gb^{(\star k)}(t) \star \Gb(t) \star \lambdab^{(0)}(t)  
    = \Gb^{(\star k+1)}(t) \star  \lambdab^{(0)}(t) 
 \end{equation}
    
\item
Lemma \ref{lem:laplace} is responsible for finding a closed form for $\widehat{\Gb}^{(\star k)}(z)$ and thus is not affected by a time-varying exogenous intensity. It remains unchanged. 
\item
Theorem \ref{theo:lin_rel} derives the instantaneous average intensity $\mub(t)$ and, therefore, needs to be updated accordingly using the modified Lemma~\ref{lem:br_intensity}:
\begin{equation}
    \mub(t) = \Psib(t) \star \lambdab^{(0)}(t) = \rbr{ e^{(\boldsymbol{A}-\omega \boldsymbol{I})t} +
    \omega (\boldsymbol{A}-\omega \boldsymbol{I})^{-1} ( e^{(\boldsymbol{A}-\omega \boldsymbol{I})t} - \boldsymbol{I} ) } \star \lambdab^{(0)}(t).
\end{equation}

\end{itemize}

Many simple parametrized incentive functions, such as exponential incentives $\lambdab^{(0)}(t)=\lambdab^{(0)} \exp(-\alpha t)$ with constant decay $\alpha$ or constant incentives within a 
window $\lambdab^{(0)}(t) = \lambdab^{(0)} \Ib [t_1 < t < t_2]$, for a fixed window $[t_1, t_2]$, result in linear closed form expressions between the exogenous event intensity and the expected 
overall intensity.
Nonparameteric functions result in a non-closed form expression, however, we still benefit the fact that the mapping from $\lambdab^0(t)$ to $\mub(t)$ is linear, and hence the activity shaping
problems can still be cast as convex optimization problems. 
In this case, the optimization can still be done via functional gradient descent (or variational calculus), though with some additional challenge to tackle.

There are many other interesting venues for future work. For example, considering competing incentives, discovering the branching structure and using it explicitly to shape the activities, exploring other 
possible kernel functions or even learning them using non-parametric methods.

\section*{Acknowledgement}
The research was supported in part by NSF/NIH BIGDATA 1R01GM108341, NSF IIS-1116886, NSF CAREER IIS-1350983 and a Raytheon Faculty Fellowship to L.S.

\clearpage
\newpage

{\small
\bibliographystyle{unsrt}

}

\clearpage
\newpage

\begin{appendix}

\section{Proofs}
\label{append:proofs}
\vspace{2mm}

\setcounter{theorem}{0}

\begin{lemma}
   \label{lem:br_intensity}
  $\mub^{(k)}(t) = \Gb^{(\star k)}(t)\,  \lambdab^{(0)}$.
\end{lemma}
\begin{proof}
  We will prove the lemma by induction. For generation $k=0$, $\mub^{(0)}(t) =  \EE_{\Hcal_t}[\lambdab^{(0)}] = \Gb^{(\star 0)}(t) \lambdab^{(0)}$. Assume the relation holds for generation $k$: $\mub^{(k)}(t) = \Gb^{(\star k)}(t) \lambdab^{(0)}$. Then for generation $k+1$, we have
  $
    \mub^{(k+1)}(t) = \EE_{\Hcal_t} [\int_{0}^t \Gb(t-s)\, d\Nb^{(k)}(s)]
    = \int_{0}^t \Gb(t-s)\, \EE_{\Hcal_t} [d\Nb^{(k)}(s)].
  $
  By definition $\EE_{\Hcal_t} [d\Nb^{(k)}(s)] = \EE_{\Hcal_{s-}}[\EE[d\Nb^{(k)}(s)|\Hcal_{s-}]] = \EE_{\Hcal_{s-}} [\lambdab^{(k)}(s)\, ds] = \mub^{(k)}(s)\, ds$, then substitute it in and we have
  \begin{align*}
    \mub^{(k+1)}(t) = \int_{0}^t \Gb(t-s)\, \Gb^{(\star k)}(s)\, \lambdab^{(0)} \, ds
    = \Gb^{(\star k+1)}(t)  \lambdab^{(0)},
  \end{align*}
  which completes the proof.
\end{proof}

\begin{lemma}
   \label{lem:laplace}
  $\widehat{\Gb}^{(\star k)}(z) = \int_0^\infty \Gb^{(\star k)}(t) \, dt = \frac{1}{z}
  \cdot \frac{\Ab^k}{(z+\omega)^k}$
\end{lemma}
\begin{proof}
  We will prove the result by induction on $k$. First, given our choice of exponential kernel, $\Gb(t) = e^{-\omega t}\boldsymbol{A}$, we have that $\widehat{\Gb}(z) = \frac{1}{z+w} \boldsymbol{A}$. Then for $k=0$, $\Gb^{(\star 0)}(t) = \Ib$ and
  $
    \widehat{\Ib}(z) =\int_0^\infty e^{-zt} \Ib\, dt = \frac{1}{z} \Ib.
  $
  Now assume the result hold for a general $k-1$, then $\widehat{\Gb}^{(\star k-1)}(z)=\frac{1}{z} \cdot \frac{\Ab^{k-1}}{(z+\omega)^{k-1}}$. Next, for $k$, we have
  $\widehat{\Gb}^{(\star k)}(z) = \int_0^\infty e^{-zt} \rbr{\Gb(t) \star \Gb^{(\star k-1)}(t)} dt = \widehat{\Gb}(z) \widehat{\Gb}^{(\star k-1)}(z)$, which is
  $
    \rbr{\frac{1}{z+\omega} \Ab} \rbr{\frac{\Ab^{k-1}}{(z+\omega)^{k-1}} \cdot \frac{1}{z}} = \frac{1}{z} \cdot \frac{\Ab^k}{(z+\omega)^k},
  $
  and completes the proof.
\end{proof}

\begin{theorem}
   \label{theo:lin_rel}
  $
    \mub(t) = \Psib(t) \lambdab^{(0)} = \rbr{ e^{(\boldsymbol{A}-\omega \boldsymbol{I})t} +
    \omega (\boldsymbol{A}-\omega \boldsymbol{I})^{-1} ( e^{(\boldsymbol{A}-\omega \boldsymbol{I})t} - \boldsymbol{I} ) } \lambdab^{(0)}.
  $
\end{theorem}

\begin{proof}
  We first compute the Laplace transform $\widehat{\Psib}(z):=\int_0^\infty e^{-z t}\Psib(t)\, dt$. Using lemma~\ref{lem:laplace}, we have
  $$
    \widehat{\Psib}(z) = \frac{1}{z} \sum_{i=0}^{\infty} \frac{\boldsymbol{A}^i}{(z+w)^i} = \frac{(z+w)}{z} \sum_{i=0}^{\infty} \frac{\boldsymbol{A}^{i+1}}{(z+w)^i}
  $$
  Second, let $\widehat{\Fb}(z) := \sum_{i=0}^{\infty} \frac{\boldsymbol{A}^i}{z^{i+1}}  $ and its inverse Laplace transform be $\Fb(t) = \int_0^\infty e^{zt} \widehat{\Fb}(z)\,dz =  \sum_{i=0}^{\infty} \frac{(\Ab t)^i}{i!} = e^{\Ab t}$, where $e^{\Ab t}$ is a matrix exponential.
  Then, it is easy to see that
  $
    \widehat{\Psib}(z)  = \frac{(z+w)}{z} \widehat{\boldsymbol{F}} (z+w) = \widehat{\boldsymbol{F}} (z+w) + \frac{w}{z} \widehat{\boldsymbol{F}} (z+w).
  $
  Finally, we perform inverse Laplace transform for $\widehat{\Psib}(z)$, and obtain
  $
  \Psib(t) = e^{(\boldsymbol{A}-\omega \boldsymbol{I})t} +  \omega \int_0^t e^{(\boldsymbol{A}-\omega \boldsymbol{I})s} ds
  = e^{(\boldsymbol{A}-\omega \boldsymbol{I})t} +
  \omega (\boldsymbol{A}-\omega \boldsymbol{I})^{-1} \left(  e^{(\boldsymbol{A}-\omega \boldsymbol{I})t} - \boldsymbol{I} \right),
  $
  where we made use of the property of Laplace transform that dividing by $z$ in the frequency domain is equal to an integration in time domain, and
  ${F} (z+w)= e^{-\omega t}e^{\boldsymbol{A}t} = e^{(\boldsymbol{A}-\omega \boldsymbol{I})t}$.
\end{proof}

\begin{corollary}
  $\mub =  \left(\mathbf{I} - \boldsymbol{\Gamma}\right)^{-1} \lambdab^{(0)} = \lim_{t \rightarrow \infty} \Psib(t) \, \lambdab^{(0)}$.
\end{corollary}
\begin{proof}
  If the process is stationary, the spectral radius of $\boldsymbol{\Gamma} = \frac{\boldsymbol{A}}{w}$ is smaller than 1, which implies that all eigenvalues of $\Ab$ are smaller than $\omega$ in magnitude.
  Thus, all eigenvalues of $\Ab - \omega \Ib$ are negative.
  Let $\Pb \Db \Pb^{-1}$ be the eigenvalue decomposition of $\Ab - \omega \Ib$, and all the elements (in diagonal) of $\Db$ are negative. Then based on the property of matrix exponential, we have $e^{(\boldsymbol{A}-\omega \boldsymbol{I})t} = \Pb e^{\Db t} \Pb^{-1}$.
  As we let $t \rightarrow \infty$, the matrix $e^{\Db t}\rightarrow \zero$ and hence $e^{(\boldsymbol{A}-\omega \boldsymbol{I})t} \rightarrow \zero$. Thus
  $
    \lim_{t \rightarrow \infty}  \Psib(t) =  -\omega (\boldsymbol{A}-\omega \boldsymbol{I})^{-1},
  $
  which is equal to $(\Ib - \boldsymbol{\Gamma})^{-1}$, and completes the proof.
\end{proof}

\section{More on Experimental Setup}
\label{append:evaluation}
Table~\ref{tab:URLuses} shows the number of adopters and usages for the six different URL shortening services.
\begin{table} [t]
\small
\centering
\begin{tabular}{|c|c|c|}
 \hline
Service & $\#$ adopters & $\#$ usages \\ \hline \hline
Bitly & $55{,}883$ & $5{,}046{,}710$ \\ \hline
TinyURL & $46{,}577$ & $1{,}682{,}459$ \\ \hline
Isgd & $28{,}050$ & $596{,}895$ \\ \hline
TwURL & $15{,}215$ & $197{,}568$ \\ \hline
SnURL & $4{,}462$ & $41{,}823$ \\ \hline
Doiop & $88$ & $643$ \\ \hline
\end{tabular}
\caption{$\#$ of adopters and usages for each URL shortening service.}
\label{tab:URLuses}
\end{table}
It includes a total of 7,566,098 events (adoptions) during the 8-month period.

In the following, we describe the considered baselines proposed to compare to our approach for i) the capped activity maximization; ii) the minimax activity shaping; and iii) the least-squares activity shaping problems.

For \textbf{\emph{capped activity maximization}} problem, we consider the following four heuristic baselines: 

 \begin{itemize}

 \item

XMU allocates the budget based on users' current activity. In particular, it assigns the budget to each of the half top-most active users proportional to their average activity, $\boldsymbol{\mu}(t)$, computed from the inferred parameters.

\item

WEI  assigns positive budget to the users proportionally to their sum of out-going influence  ($\sum_u a_{uu'}$). This heuristic allows us (by comparing its results to CAM) to understand the effect of considering the whole network with respect to only consider the direct (out-going) influence.

\item

 DEG assumes that more central users, \ie, more connected users, can leverage the total activity, therefore, assigns the budget to the more connected users proportional to their degree in the network.

\item

PRK sorts the users according to the their pagerank in the weighted influence network ($\Ab$) with the damping factor set to $0.85 \%$, and assigns the budget to the top users proportional their pagerank value.

 \end{itemize}

In order to show how network structure leverages the \textbf{\emph{minimax activity shaping}} we implement following four baselines: 

\begin{itemize}

\item

UNI allocates the total budget equally to all users.

\item

MINMU divides uniformly the total budget among half of the users with lower average activity $\boldsymbol{\mu}(t)$, which is computed from the inferred parameters.

\item

LP finds the top half of least-active users in the current network and allocates the budget such that after the assignment the network has the highest minimum activity possible. This method uses linear programming to learn exogenous activity of the users, but, in contrast to the proposed method, does not consider the network and propagation of adoptions.

\item

GRD  finds the user with minimum activity, assigns a portion of the budget, and computes the resulting $\boldsymbol{\mu}(t)$. It then repeats the process to incentivize half of users.

 \end{itemize}

We compare \textbf{\emph{least-square activity shaping}} with the following baselines:

 \begin{itemize}

 \item

PROP shapes the activity by allocating the budget proportional to the desired shape, \ie, the shape of the assignment is similar to the target shape.

 \item

 LSGRD greedily finds the user with the highest distance between her current and target activity, assigns her a budget to reach her target, and proceeds this way to consume the whole budget.

 \end{itemize}

Each baseline relies on a specific property to allocate the budget (\eg~connectedness in DEG). However, most of them face two problems: The first one is how many users to incentivize and the second one is how much should be paid to the selected users. They usually rely on heuristics to reveal these two problems (\eg~allocating an amount proportional to that property and/or to the top half users sorted based on the specific property). In contrast, our framework  is comprehensive enough to address those difficulties based on well-developed theoretical basis.  This key factor accompanied with the appropriate properties of Hawkes process for modeling social influence (\eg~mutually exciting) make the proposed method the best.

\clearpage
\newpage

\section{Temporal Properties}
\label{append:temp}
For the experiments on simulated objective function and held-out data we have estimated intensity from the events data. In this section, we will see how this empirical intensity resembles the theoretical intensity.
We generate a synthetic network over 100 users. For each user in the generated network,  we uniformly sample from $[0, 0.1]$ the exogenous intensity, and the endogenous parameters $a_{uu'}$ are uniformly sampled from $[0,0.1]$.  A bandwidth $\omega =1$ is used in the exponential kernel. Then,  the intensity is estimated empirically by dividing the number of events by the length of the respective interval. 

We compute the mean and variance of the empirical activity for 100 independent runs. As illustrated in  Figure \ref{fig:convergence}, the average empirical intensity (the blue curve) clearly follows the theoretical instantaneous intensity (the red curve) but, as expected, as we are further from the starting point (\ie, as time increases), the standard deviation of the estimates (shown in the whiskers) increases. Additionally, the green line shows the average stationary intensity. As it is expected, the instantaneous intensity tends to the stationary value when the network has been run for sufficient long time.

\begin{figure}[h]
\centering
                \includegraphics[width=0.33\columnwidth]{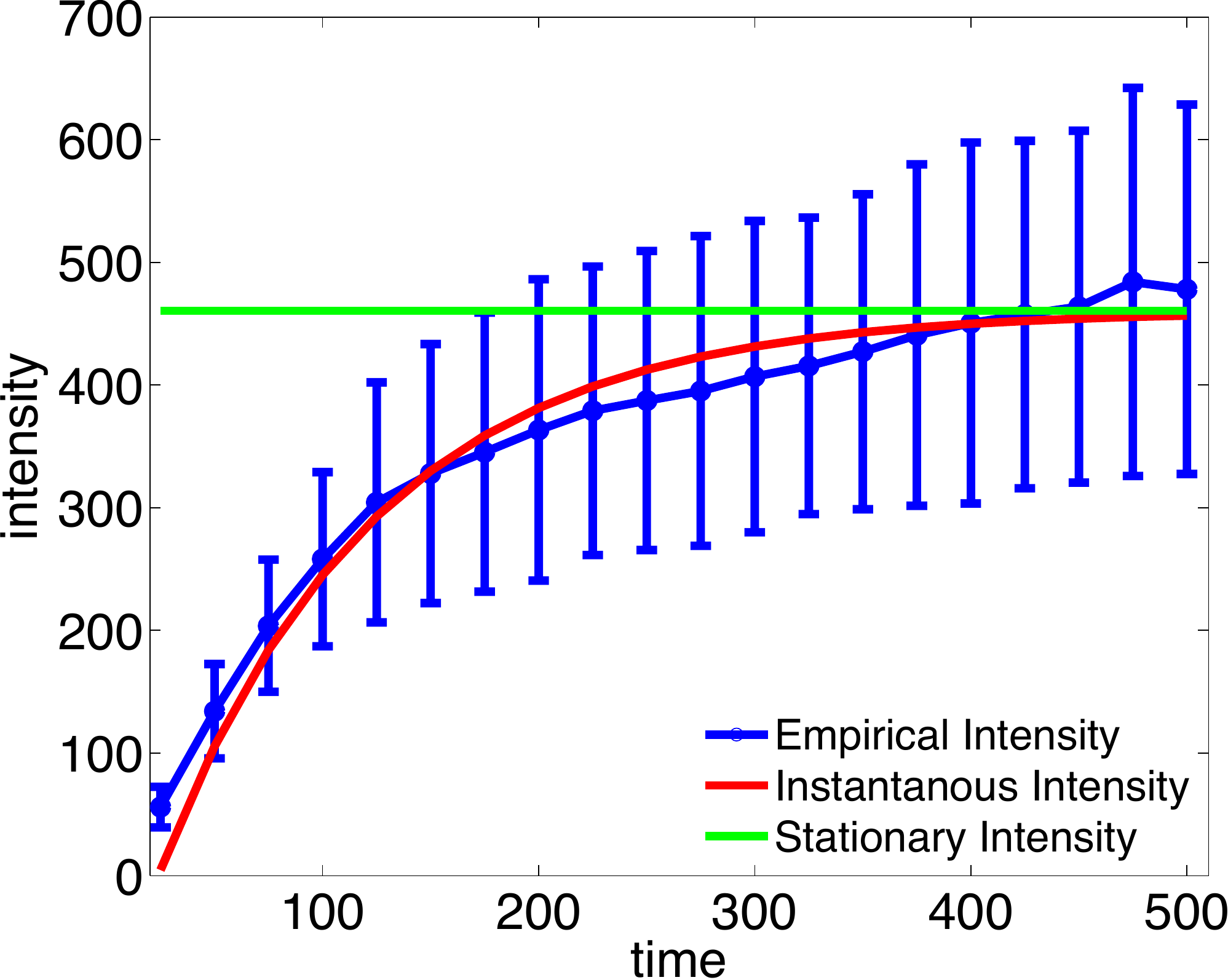}
               \caption{Evolution in time of empirical and theoretical intensity.}
                \label{fig:convergence}
\end{figure}

\section{Visualization of Least-squares Activity Shaping }

To get a better insight on the the activity shaping problem we visualize the \emph{least-squares activity shaping} results for the 2K and 60K datasets. Figure \ref{fig:results5} shows the result of activity shaping at $t=1$ targeting the same shape as in the experiments section. The red line is the target shape of the activity and the blue curve correspond to the activity profiles of users after incentivizing computed via theoretical objective. It is clear that the resulted activity behavior resembles the target shape.

\label{append:visual}
\begin{figure*} [h]
        \centering
        \begin{tabular}{cc}
                \includegraphics[width=0.33\textwidth]{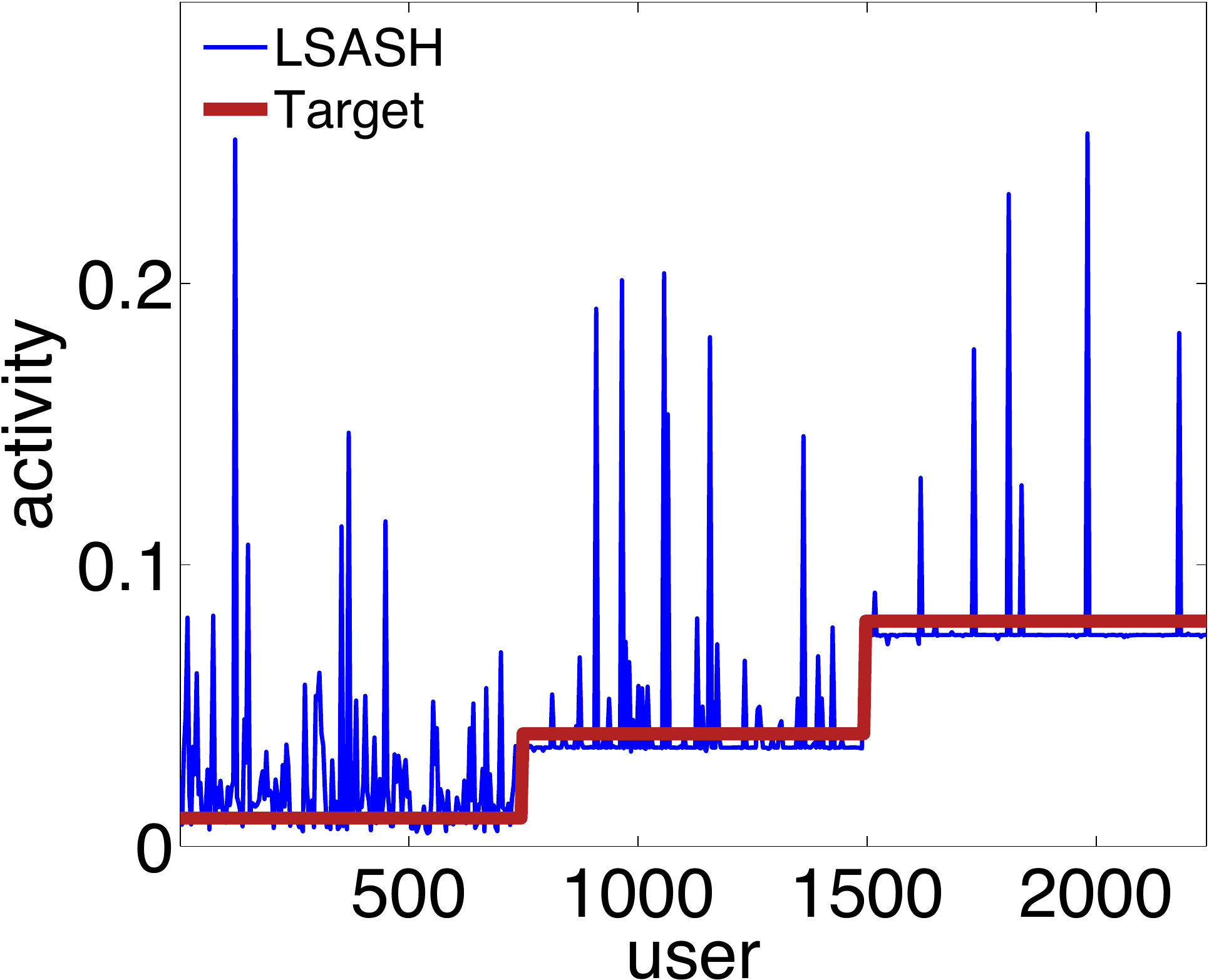} \label{fig:results5:lsq2k} &
                \includegraphics[width=0.33\textwidth]{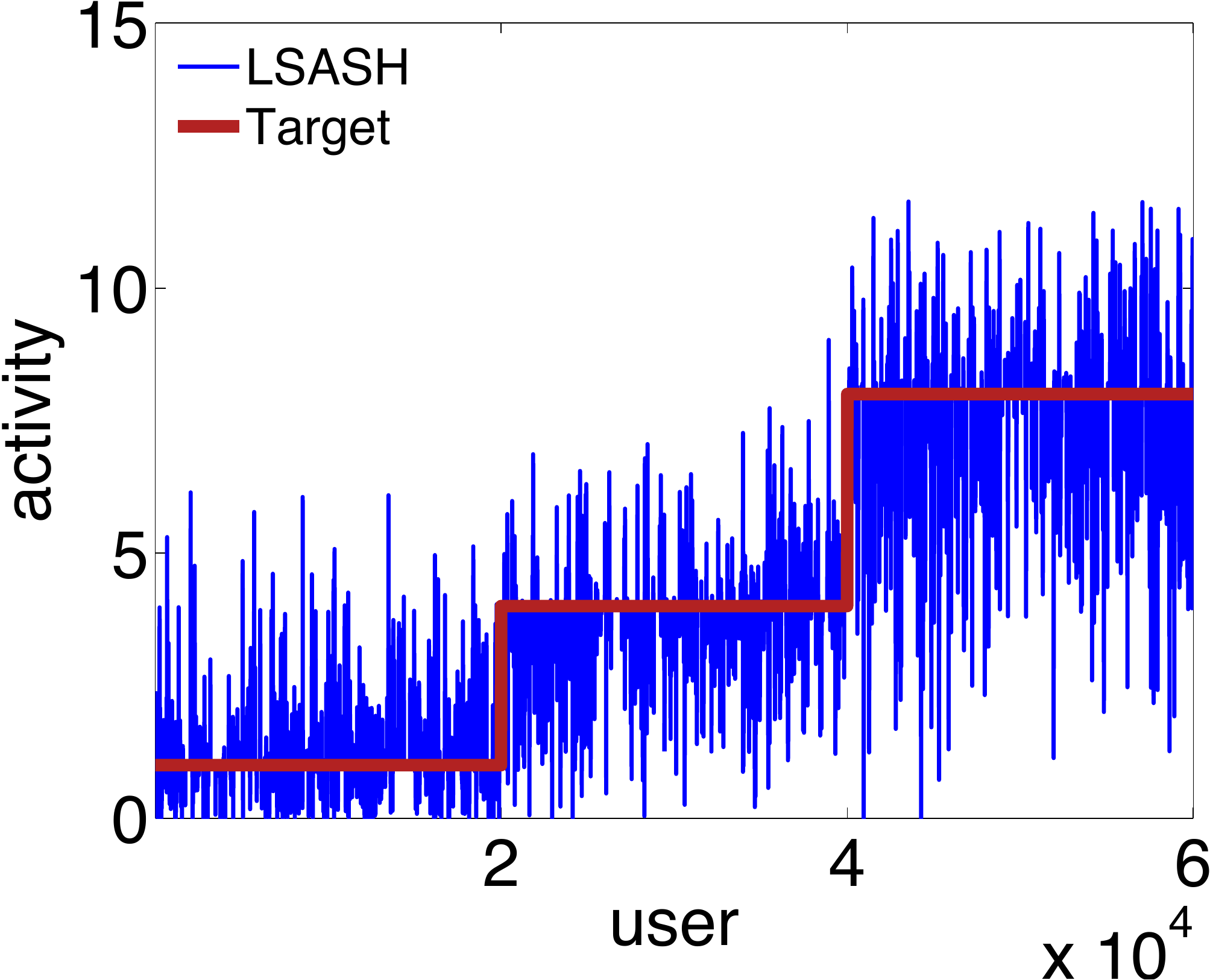} \label{fig:results5:lsq60k} \\
         	(a) Capped activity maximization. &
		(b) Minimax activity shaping.
	\end{tabular}
        \caption{Activity shaping results.}\label{fig:results5}
\end{figure*}

\end{appendix}

\end{document}